\documentclass[10pt,twocolumn]{article}
\usepackage{causalquery}
\title{Causal Query Compression for Lindblad Dynamics\\Optimal Queries and Nearly Linear Local Simulation}
\author{Jacob Kitchen}
\date{September 25, 2026}
\hypersetup{pdftitle={Causal Query Compression for Lindblad Dynamics: Optimal Queries and Nearly Linear Local Simulation},pdfauthor={Jacob Kitchen}}
\begin{document}
\twocolumn[\begin{@twocolumnfalse}
\maketitle
\begin{abstract}
We construct a causal query compiler for time-dependent Lindblad dynamics and give
two applications with different input models. Under coherent block-encoding access to
the Hamiltonian and individual jump operators, a Lipschitz generator can be simulated
to diamond error \(\eps\) using
\[
 O\!\left(1+\tau+
 \frac{\log(1/\eps)}{\log(e+\log(1/\eps)/\tau)}\right),
 \qquad \tau=T\left(\alpha_H+\sum_\mu\alpha_\mu^2\right),
\]
controlled queries, independently of the time-grid size. For \(\tau\ge1\), this
query bound is optimal in the worst case. Weak active entry and clean return give a
factorial history tail, and two reuse lengths remove the transducer's auxiliary input.
A three-call polar correction supplies the required Lindblad cells. The compiler also
equates the minimum oracle action of bounded adaptive Markovian protocols with quantum
query complexity and the general adversary bound, up to constant factors.

For finite-range dynamics on \(n\) lattice sites with fixed local parameters and
efficient coherent matrix evaluation, we obtain
\[
 O\!\left(n(T+1)\polylog\frac{n(T+1)(1+K_t+B)}{\eps}\right)
\]
elementary gates for piecewise H\"older generators with fixed exponent, local variation
bound \(K_t\), and at most \(B\) breakpoints per term. No commutativity of the jumps
is required. A spatial decomposition on a common bath, compilation on occupied inputs,
and routing on compressed records account for evaluation, arithmetic, selection, and
environment storage. Reusing bath storage between short segments gives depth
\((T+1)\polylog X\) and space \(n\polylog X\), where
\(X=\max\{e,n(T+1)(1+K_t+B)/\eps\}\).
\end{abstract}
\end{@twocolumnfalse}
\vspace{1em}
]

\section{Introduction}

We study two resource questions for time-dependent Lindblad evolution. In the oracle
model, the Hamiltonian and individual jump operators are supplied through coherent
block encodings, and the cost is the number of controlled oracle calls. For local
lattice dynamics, the input instead consists of coherent evaluation circuits for the
local matrices, and every elementary gate is charged. Causal query compression provides
the common circuit construction; its gate implementation requires additional spatial
and bath-occupation estimates.

\subsection{Query complexity and oracle action}

Given block encodings of \(H(t)\) and \(L_\mu(t)\), the objective is to approximate
the channel from time \(0\) to \(T\) in diamond norm. Write
\(\Lambda=\alpha_H+\sum_\mu\alpha_\mu^2\) and \(\tau=T\Lambda\), using the
supplied normalization bounds. For a Lipschitz generator and \(0<\eps<1/8\),
Theorem~\ref{thm:lindblad} gives the query bound
\begin{equation}
 O\!\left(1+\tau+
 \frac{\log(1/\eps)}{\log\!\left(e+\log(1/\eps)/\tau\right)}\right).
 \label{eq:intro-bound}
\end{equation}
For \(\tau=0\), no query is needed. For \(\tau\ge1\), the worst-case lower bound
has the same order, already for Hamiltonian instances with no jump operators.

Nearly linear dependence on normalized time is known for general Lindblad
simulation~\cite{CleveWang17,LiWang23,HeEtAl24}, with a multiplicative polylogarithmic
precision factor. Additive precision costs have been obtained for scalar total jump
rates and unitary jump operators~\cite{BorrasMarvian26,ShangEtAl25}. Time dependence
also requires a mesh that resolves the generator's variation. At fixed normalization,
this mesh may contain arbitrarily many intervals; applying a constant-query channel
at each interval would incur a query cost proportional to that number.

Chen, Gao, Wang, and Zhou~\cite{CGWZ26} resolve this issue for Hamiltonian simulation
under coherent time-indexed access. Their chronological transducer has a factorial
history tail. We give local sufficient conditions for the same tail and verify them
for a short Lindblad circuit. This uses the transducer framework of Belovs, Jeffery,
and Yolcu~\cite{BJY24}.

\begingroup
\setlength{\fboxsep}{6pt}
\medskip\noindent
\fbox{\begin{minipage}{\dimexpr\linewidth-2\fboxsep-2\fboxrule\relax}
\small
\textbf{Query model.} Time and jump labels can be selected coherently. Controlled
block-encoding calls and their adjoints are supplied on the chosen grid. One unit
of cost is a master SELECT call, implemented with \(O(1)\) supplied queries.
If the encodings have only classical indices, building SELECT may require a number
of calls proportional to the grid size.
\end{minipage}}\par\medskip
\endgroup

For cells of weights \(w_j\), the compiler uses small amplitudes both at active entry
and at clean return. A clean path pays both factors; a nonclean path may pay only one,
but its fresh environment is retained and becomes orthogonal to the clean output.
Keeping interfering histories together in a generating function gives
\begin{equation}
 \norm{A^{mq}\Gamma P_{\rm in}}
 \le C\left(\frac{C\Omega}{q}\right)^{q/2},
 \qquad \Omega=\sum_jw_j.
 \label{eq:intro-tail}
\end{equation}
Here \(A\) is the private transducer block and \(m\) bounds the slots per cell.
Combining reuse lengths \(K\) and \(3K\), with coefficients \(-1/2\) and \(3/2\),
removes the auxiliary input. One rectangular amplification step restores the isometry.
The three-call polar correction of an Euler column supplies Lindblad cells with
weight \(h\Lambda\), so the total weight is \(T\Lambda\) on every mesh.

The same compiler applies to bounded adaptive Markovian protocols with measurements,
classical feedback, coherent access, and uniform temporal regularity. If \(MQ(f)\)
is the minimum oracle action for computing a finite function in this model, then
Corollary~\ref{cor:adversary} gives
\[
 MQ(f)=\Theta(Q(f))=\Theta(\operatorname{Adv}^{\pm}(f)).
\]
The environments and histories are retained coherently during compilation and are
discarded only when passing to the accessible output channel.

\subsection{Nearly linear local gate complexity}

A small query count alone does not bound the work of traversing a fine time mesh or
storing its environments. We implement that traversal for finite-range generators on
\(n\) lattice sites, in any fixed spatial dimension. The local dimensions, interaction
range, term degree, jump count, and normalization bounds are fixed. The generator of
each term is piecewise H\"older with fixed exponent, variation bound \(K_t\), and at
most \(B\) breakpoints. With
\[
 X=\max\{e,n(T+1)(1+K_t+B)/\eps\},
\]
Theorem~\ref{loc:thm:main} gives
\begin{equation}
 G=O\bigl(n(T+1)\polylog X\bigr)
 \label{loc:eq:intro-bound}
\end{equation}
one- and two-qubit gates, depth \((T+1)\polylog X\), and \(n\polylog X\) working
qubits when baths are discarded between short time segments. The jumps need not
commute. No matching local gate lower bound is claimed.

\begingroup
\setlength{\fboxsep}{6pt}
\medskip\noindent
\fbox{\begin{minipage}{\dimexpr\linewidth-2\fboxsep-2\fboxrule\relax}
\small
\textbf{Local input and gate model.} Each constant-dimensional local matrix admits
polynomial-cost coherent evaluation to \(b\) bits from its time and spatial address.
Time arithmetic and coherent lookup are included. Patch selection is implemented by
multiplexing these evaluators. Controls, adjoints, coefficient preparation, routing,
and bath operations all enter the gate count.
\end{minipage}}\par\medskip
\endgroup

For Hamiltonians, overlapping regions give nearly linear space-time gate
cost~\cite{HHKL}. In a Lindblad decomposition, the middle inverse in a product
\(U_{AB}U_B^\dagger U_{BC}\) must act on the bath records left by the preceding
regional factor. Resetting those records would change the product. We keep the same
bath through the \(3^D\) passes of the spatial construction and estimate the complete
word with vacuum only at its initial input.

At finite mesh size, a replicated-system contraction groups the word by connected
interaction support. Vacuum return contributes a factor of the bin width for each
selected interaction. Inclusion--exclusion cancels a connected component unless it
reaches opposite cut colors. Bounded-degree support counting then permits patches of
\(O(\log^D X)\) sites. The occupied-input compiler extends the shared history-tail
argument to the bath states entering subsequent regional factors.

For the gate implementation, common normalizations make the singular frames and
baseline routing coefficients independent of the sampled operators. The free
transducer preserves occupation in these coordinates. We route it on stored records
by an inhomogeneous dyadic circuit and implement the collective frame changes in a
symmetric record encoding. Separate estimates control occupation at physical regional
boundaries and throughout the compiled circuit.

Low-Hamming-weight encodings occur in continuous-query simulation~\cite{Berry}.
Our routing uses an occupation-dependent extension of the dyadic identity of Chen,
Gao, Ji, Li, Wang, and Zhou~\cite{ChenGates}. The replicated-system viewpoint is related
to quantum regression~\cite{Blocher}, and the spatial layout follows the overlapping
geometry of~\cite{HHKL}. The smooth-Gaussian-environment algorithm of Yu, Li, Cirac,
and Trivedi~\cite{Yu} requires commuting jumps in its Markovian case.
Commutator-sensitive product formulas in~\cite{Wang} improve system-size dependence;
their extrapolation gives polylogarithmic precision cost for observable estimation,
while full-channel precision cost remains polynomial.

\paragraph{Independent concurrent work.}
Chen, Gao, Wang, and Zhou~\cite{CGWZL26} give an independent query-optimal Lindblad
simulation using a different weak-cell construction. Their construction uses an
exactly trace-preserving and completely positive one-query interpolation based on
global jump-SELECT access. Here the three-call Euler correction provides the primitive
port bounds and the fixed local dilation used for the gate implementation.

\paragraph{Organization.}
Sections~\ref{sec:compiler} and~\ref{sec:slice} establish the shared compiler and
Lindblad cell. Sections~\ref{sec:time-dependent} and~\ref{sec:adaptive} prove the
general query and oracle-action results. Sections~\ref{loc:sec:cells}--\ref{loc:sec:assembly}
give the local gate construction, from occupied-input compilation and the common-bath
decomposition to routing, record cutoffs, and the total resource bounds.

\section{A compiler for causal query programs}
\label{sec:compiler}

Consider a query controlled by a qubit with a small amplitude in its active state. What happens when the control rotation is undone? The clean output acquires a second small factor; the orthogonal output does not. We can see the resulting \(q/2\) exponent in a circuit with just one query.

\paragraph{A one-query example.}
Let \(U\) be a unitary oracle and \(0\le w\le1\). Rotate a fresh qubit by
\[
 R_w=\begin{pmatrix}\sqrt{1-w}&-\sqrt w\\
                    \sqrt w&\sqrt{1-w}\end{pmatrix},
\]
apply \(U\) controlled on that qubit being one, and apply \(R_w^\dagger\). Retain the qubit as part of the output. Starting from \(\ket0\ket\psi\), the clean and nonclean output operators are
\[
 \begin{aligned}
 C_w&=(1-w)I+wU,\\
 N_w&=\sqrt{w(1-w)}(U-I).
 \end{aligned}
\]
In \(C_w\), the two rotations each contribute \(\sqrt w\) to the oracle term. Only one such factor occurs in \(N_w\). There is a nonclean contribution even when the query is replaced by its inactive projector: it is \(-\sqrt{w(1-w)}I\). We will include this case in the cell assumptions.

To see the exponent in the tail bound, replace the controlled \(U\) by \(zU\), and write \(T_w(z)=(C_w(z),N_w(z))^{\mathsf T}\). Orthogonality of the two output sectors and unitarity of \(U\) give the exact identity
\[
 \begin{aligned}
 T_w(z)^\dagger T_w(z)
 &=C_w(z)^\dagger C_w(z)+N_w(z)^\dagger N_w(z)\\
 &=(1-w+w|z|^2)I.
 \end{aligned}
\]
For a product \(\mathcal V(z)\) of these cells, each with a fresh retained qubit, set \(\Omega=\sum_jw_j>0\). Then \(\norm{\mathcal V(z)}\le e^{\Omega R^2/2}\) on \(|z|=R\). Choosing \(R^2=q/\Omega\) in Cauchy's coefficient estimate yields, for integers \(q\ge1\),
\[
 \norm{[z^q]\mathcal V(z)}
 \le R^{-q}e^{\Omega R^2/2}
 =\left(\frac{e\Omega}{q}\right)^{q/2}.
\]
At the chosen radius, \(R^{-q}=(\Omega/q)^{q/2}\), and the exponential contributes \(e^{q/2}\). This gives the power \(q/2\) above. With up to \(m\) slots per cell, the balance is between \(R^{-mq}\) and \(e^{O(\Omega R^{2m})}\), so the power is again \(q/2\). For a full query history, the last hit can occur before its cell is finished. The estimate in Section~\ref{subsec:history-generating} includes these unfinished paths.

To record successive hits, we give every oracle slot its own private port. The transducer built from these ports keeps their chronological order, so its matrix powers can be used to count query histories.

\subsection{Exposing the query history}

Let \(\mathcal W\) be the public workspace of a finite unitary query program
\begin{equation}
 V=F_M\widetilde O_MF_{M-1}\cdots F_1\widetilde O_1F_0,
 \qquad \widetilde O_t=Q_t+O_tP_t.
 \label{eq:query-program}
\end{equation}
The gates \(F_t\) and orthogonal projectors \(P_t\) are oracle independent, \(Q_t=I-P_t\), and \(O_t\) is unitary on the active subspace \(P_t\mathcal W\). We assume that the direct sum of the \(O_t\)'s can be implemented with one coherent master query. If forward and inverse queries are charged separately, a constant number suffices throughout.

For each slot, introduce a private copy \(\mathcal L_t\) of \(P_t\mathcal W\), with inclusion \(\iota_t:\mathcal L_t\to\mathcal W\). The operator
\begin{equation}
 X_t=\begin{pmatrix}Q_t&\iota_t\\ \iota_t^\dagger&0\end{pmatrix}
 \label{eq:capture-release}
\end{equation}
acts on \(\mathcal W\oplus\mathcal L_t\). Since \(\iota_t\iota_t^\dagger=P_t\), \(\iota_t^\dagger\iota_t=I\), and \(Q_t\iota_t=0\), it satisfies \(X_t=X_t^\dagger\) and \(X_t^2=I\). Thus it swaps the active public component with the private port and fixes the inactive component.

Extend each \(X_t\) by the identity on the other private ports. On \(\mathcal W\oplus\mathcal L\), where \(\mathcal L=\bigoplus_{t=1}^M\mathcal L_t\), define
\begin{align}
 S^\circ&=(F_M\oplus I)X_M\cdots(F_1\oplus I)X_1(F_0\oplus I),\notag\\
 \overline O_t&=\iota_t^\dagger O_t\iota_t,
 \qquad \widehat O=\bigoplus_{t=1}^M\overline O_t,
 \qquad S=S^\circ(I\oplus\widehat O).
 \label{eq:cut-open-transducer}
\end{align}
The routing circuit \(S^\circ\) is oracle free, and \(S\) is a unitary using one master query. Write its public--private decomposition as
\begin{equation}
 S=\begin{pmatrix}D&C\\B&A\end{pmatrix}.
 \label{eq:Sblocks}
\end{equation}
Read \(S\) by its input columns: \(D\) and \(B\) map public inputs to public outputs and private ports; \(C\) and \(A\) do the same for private inputs.

Restrict the inputs by the initialization projector \(P_{\rm in}\), which fixes all circuit ancillas. For \(x\in\operatorname{ran}P_{\rm in}\), define
\begin{align}
 \psi_0(x)&=F_0x,\notag\\
 \gamma_t(x)&=\iota_t^\dagger P_t\psi_{t-1}(x),\notag\\
 \psi_t(x)&=F_t\bigl(Q_t\psi_{t-1}(x)+\iota_t\overline O_t\gamma_t(x)\bigr),
 \qquad 1\le t\le M,
 \label{eq:chronological-states}
\end{align}
and define the coherent query history
\begin{equation}
 \Gamma x=\bigoplus_{t=1}^M\gamma_t(x),
 \qquad
 \norm{\Gamma x}^2=\sum_{t=1}^M\norm{P_t\psi_{t-1}(x)}^2.
 \label{eq:gamma-definition}
\end{equation}
If \(x\) is normalized, the summand for a slot is its active-subspace probability. Adding over slots gives the query weight \(\norm{\Gamma x}^2\).

We can apply \(S\) to \(x\oplus\Gamma x\) without normalizing that vector. Just before the swap at slot \(t\), its public component is \(\psi_{t-1}(x)\). The private component at that slot has already received its oracle call and is \(\overline O_t\gamma_t(x)\). Thus
\[
 X_t\binom{\psi_{t-1}(x)}{\overline O_t\gamma_t(x)}
 =\binom{Q_t\psi_{t-1}(x)+\iota_t\overline O_t\gamma_t(x)}{\gamma_t(x)}.
\]
Induction over the slots gives the graph relation
\begin{equation}
\begin{gathered}
 S\binom{P_{\rm in}}{\Gamma P_{\rm in}}
 =\binom{VP_{\rm in}}{\Gamma P_{\rm in}},\\
 \begin{cases}
 BP_{\rm in}+A\Gamma P_{\rm in}=\Gamma P_{\rm in},\\
 DP_{\rm in}+C\Gamma P_{\rm in}=VP_{\rm in}.
 \end{cases}
 \end{gathered}
 \label{eq:graph}
\end{equation}
This is a transducer implementation of \(V\), with auxiliary state \(\Gamma x\).

\begin{lemma}[Strict chronology]
\label{lem:strict-chronology}
In slot order, \(A\) is strictly lower triangular. Consequently \(A^M=0\), and
\begin{equation}
\begin{aligned}
 \Gamma P_{\rm in}&=\sum_{r=0}^{M-1}A^rBP_{\rm in},\\
 VP_{\rm in}&=\left(D+C\sum_{r=0}^{M-1}A^rB\right)P_{\rm in}.
 \end{aligned}
 \label{eq:graph-expansion}
\end{equation}
For \(r\ge1\), the term \(CA^{r-1}B\) sums all public paths with exactly \(r\) active hits.
\end{lemma}

\begin{proof}
After release from port \(t\), the next possible capture is at some slot \(s\) later in the circuit. The slots between release and capture contribute inactive projectors, giving
\begin{equation}
 [A]_{s,t}=\begin{cases}
 \iota_s^\dagger P_s\mathcal Q_{s\leftarrow t}\iota_t\overline O_t,&s>t,\\
 0,&s\le t,
 \end{cases}
 \label{eq:Achronological}
\end{equation}
where
\[
 \mathcal Q_{s\leftarrow t}
 =F_{s-1}Q_{s-1}\cdots F_{t+1}Q_{t+1}F_t,
 \qquad \mathcal Q_{t+1\leftarrow t}=F_t.
\]
Nilpotence now follows from chronology. Solving the first block equation in \eqref{eq:graph} gives the finite series for \(\Gamma\), and the second gives the series for \(V\). Each factor of \(A\) advances to one more active slot. An inverse oracle obeys the same ordering.
\end{proof}

\subsection{Removing the auxiliary state}

Use the construction of~\cite{BJY24}: distribute the public input over \(N\) bins with equal amplitudes, and pass a single private rail through the bins in order. At bin \(\ell\), apply \(S\) to that bin's public component together with the rail. Feed its private output into the next bin. These are unitary operations on a direct sum, with no copying. We write \(R_N\) for the public block after projecting the bin register onto the same uniform state used at initialization.

When the private input is zero, its amplitude before bin \(\ell\) satisfies
\begin{equation}
\begin{gathered}
 z_0=0,\qquad z_{\ell+1}=Az_\ell+\frac{Bx}{\sqrt N},\\
 z_\ell=\frac{(I-A^\ell)\Gamma x}{\sqrt N}.
 \end{gathered}
 \label{eq:reuse-recurrence}
\end{equation}
The public output from that bin is \(Dx/\sqrt N+Cz_\ell\). Averaging over the bins therefore yields
\begin{equation}
 R_Nx=Dx+\sum_{r=0}^{N-2}\left(1-\frac{r+1}{N}\right)CA^rBx.
 \label{eq:reuse-public-block}
\end{equation}
Equivalently,
\begin{equation}
\begin{gathered}
 (V-R_N)P_{\rm in}=CG_N(A)\Gamma P_{\rm in},\\
 G_N(z)=\frac1N\sum_{r=0}^{N-1}z^r.
 \end{gathered}
 \label{eq:reuse-identity}
\end{equation}

The coefficients below degree \(K\) in \(G_K\) and \(G_{3K}\) are constant. They can all be cancelled while preserving the target coefficient, since
\begin{equation}
 -\frac12G_K(z)+\frac32G_{3K}(z)=z^KG_{2K}(z).
 \label{eq:reuse-filter-polynomial}
\end{equation}
Define \(\widetilde R_K=-R_K/2+3R_{3K}/2\). The error becomes
\begin{equation}
 (V-\widetilde R_K)P_{\rm in}=CA^KG_{2K}(A)\Gamma P_{\rm in}.
 \label{eq:filtered}
\end{equation}
Since \(A\) and \(C\) are contractions and \(\norm{G_{2K}(A)}\le1\),
\begin{equation}
 \norm{(V-\widetilde R_K)P_{\rm in}}
 \le\eta_K,
 \qquad \eta_K:=\norm{A^K\Gamma P_{\rm in}}.
 \label{eq:filtered-error}
\end{equation}

The two coefficients have absolute sum two. A linear combination of the reuse circuits consequently encodes \(\widetilde R_KP_{\rm in}/2\) using at most \(3K\) controlled calls to \(S\). Pad the shorter branch with identities so both branches share the same query schedule. Let \(W\) be this unitary encoding, with right and left signal projectors
\[
 \Pi_R=P_{\rm in}\otimes\ketbra0{0}_{\rm work},
 \qquad \Pi_L=I_{\rm target}\otimes\ketbra0{0}_{\rm work}.
\]
Only the added work registers are selected on the output side; the environment of the target isometry remains part of its output. With \(R_R=2\Pi_R-I\), \(R_L=2\Pi_L-I\), one round of amplification gives
\begin{equation}
\begin{gathered}
 U_{\rm amp}=-WR_RW^\dagger R_LW,\\
 \Pi_LU_{\rm amp}\Pi_R=3X-4XX^\dagger X,\qquad
 X=\frac{\widetilde R_KP_{\rm in}}2.
 \end{gathered}
 \label{eq:rectangular-oaa}
\end{equation}

\begin{lemma}[Restoration of an approximate isometry]
\label{lem:cubic-restoration}
Suppose \(V\) is an isometry on \(\operatorname{ran}P\), \(\norm{T-VP}\le\eta\le1/8\), and a unitary \(W\) has projected block \(T/2\). The amplified unitary in \eqref{eq:rectangular-oaa}, restricted to the initialized input, is within \(4\eta\) of \(V\) with the work registers reset.
\end{lemma}

\begin{proof}
Let \(T=U|T|\) be the polar decomposition. Its singular values lie in \([1-\eta,1+\eta]\), so
\begin{equation}
 \norm{U-VP}\le\norm{|T|-I}+\norm{T-VP}\le2\eta.
 \label{eq:polar-nearby}
\end{equation}
For a singular value \(s\) of \(T\), the amplified signal amplitude in the polar direction is \(f(s)=(3s-s^3)/2\). Both the full amplified column and the polar target are unit vectors, and
\[
 2-2f(s)=(s-1)^2(s+2)\le4\eta^2.
\]
This gives distance at most \(2\eta\) on a singular vector. For a superposition, compute the squared difference operator on the input space: it is \(2I-2f(|T|)\), so the bound is again \(2\eta\). The remaining error is bounded by \eqref{eq:polar-nearby}.
\end{proof}

\begin{theorem}[Compilation from a history tail]
\label{thm:tail-compiler}
Assume \eqref{eq:graph}. If \(K\ge1\) and \(\eta_K\le1/8\), there is a circuit implementing an isometry \(\widetilde V_K\) such that
\begin{equation}
 \norm{(\widetilde V_K-V)P_{\rm in}}\le4\eta_K.
 \label{eq:tail-compiler-error}
\end{equation}
The circuit uses at most \(9K\) controlled calls to \(S\) and \(S^\dagger\), and hence \(O(K)\) controlled master queries.
\end{theorem}

\begin{proof}
The filtered block has error at most \(\eta_K\) by \eqref{eq:filtered-error}. Apply Lemma~\ref{lem:cubic-restoration}. Each occurrence of \(W\) or \(W^\dagger\) costs at most \(3K\) transducer calls, and \eqref{eq:rectangular-oaa} contains three such occurrences.
\end{proof}

\subsection{Weak cells}

Group consecutive slots into cells of at most \(m\) slots each. At cell \(j\), initialize a fresh workspace \(E_j\) in \(\ket0_j\). The cell may act on the system, this workspace, and the permitted old registers. Once it ends, keep every register in \(E_j\); no later cell may act on them. Choose a one-dimensional clean output subspace and write
\[
 E_j=\mathbb C\ket0_j\oplus E_j^\bullet.
\]
We may choose the clean output by an oracle-independent basis change. The retained registers exclude work that is uncomputed before the cell ends.

Replace the oracle slots in a prefix by their inactive projectors and call the prefix \emph{absorbed}. The next active port then selects a path's first hit. An \emph{absorbed suffix}, with the same replacement in later slots, selects its last hit. Only the fresh workspace is initialized; the joint system and older registers are unrestricted. The one-query rotation obeys the following conditions with \(a=m=1\).

\begin{definition}[Causal weak cell]
\label{def:cell}
A cell of weight \(w_j\ge0\) is \((a,m)\)-causal if it has at most \(m\) oracle slots and the following maps have operator norm at most \(a\sqrt{w_j}\):
\begin{cellconditions}
 \item the map from a clean input to any chosen active port along an absorbed prefix;
 \item the direct sum of the active query-input maps in the ordinary cell circuit;
 \item the map from a state released at any active port to the clean output along an absorbed suffix;
 \item the map from a clean input to \(E_j^\bullet\) with every oracle slot absorbed.
\end{cellconditions}
The estimates include the system's joint state with all older registers, and still hold after attaching any reference register. Along every path, gates and slot maps, selected or inactive, have operator norm at most one.
\end{definition}

At fixed \(m\), (i) implies (ii) with a larger constant: expand a query input by the location of its first selected hit. We list (ii) separately to make the bound on query-history weight explicit. Put
\begin{equation}
 \Omega=\sum_{j=1}^Jw_j.
 \label{eq:Omega}
\end{equation}
The next theorem is the estimate needed for compilation.

\begin{theorem}[Factorial history tail]
\label{thm:causal-tail}
For a program of \((a,m)\)-causal cells, there are constants \(c,C>0\), depending only on \(a,m\), such that every integer \(q\ge c\max\{\Omega,1\}\) satisfies
\begin{equation}
 \norm{A^{mq}\Gamma P_{\rm in}}
 \le C\left(\frac{C\Omega}{q}\right)^{q/2}.
 \label{eq:tail-main}
\end{equation}
\end{theorem}

\subsection{A generating function for the history}
\label{subsec:history-generating}

Replace \(Q_t+O_tP_t\) formally by \(Q_t+zO_tP_t\) for a complex variable \(z\). On expanding the product, the power \(z^r\) records exactly \(r\) active oracle applications. This polynomial is a device for summing paths in the proof.

For a whole cell with initialized workspace, write this polynomial as \(T_j(z)\), with output blocks \(C_j(z)\) on the clean subspace and \(N_j(z)\) on its complement. For the inputs to the active query ports, use the direct sum \(K_j(z)\). Only preceding slots are deformed in these maps, so their degree is at most \(m-1\).

\begin{lemma}[Local polynomial bounds]
\label{lem:local-polynomial}
Let \(R\ge1\), \(|z|=R\), and \(\rho=1+R\). Then
\begin{align}
 \norm{C_j(z)}&\le1+a^2w_j(\rho^m-1),\notag\\
 \norm{N_j(z)}&\le a\sqrt{w_j}\,\rho^m,\notag\\
 \norm{K_j(z)}&\le a\sqrt{mw_j}\,\rho^{m-1},
 \label{eq:local-polynomial-bounds}
\end{align}
and
\begin{equation}
 \norm{T_j(z)}\le
 \exp\!\left(\tfrac32a^2w_j\rho^{2m}\right).
 \label{eq:cell-growth}
\end{equation}
\end{lemma}

\begin{proof}
Fix a nonempty subset of active slots in the cell expansion. Its first slot costs a factor \(a\sqrt{w_j}\) by (i). For a clean output, (iii) gives a second factor at the last slot. The maps between them have norm at most one. Thus each pattern contributes at most \(a^2w_j\) to the clean block, or \(a\sqrt{w_j}\) if its output is nonclean.

For the empty pattern, the clean block is a contraction and the nonclean block has norm at most \(a\sqrt{w_j}\) by condition (iv). There are at most \(\binom m\ell\) patterns with \(\ell\) selected slots. Summing their norms with weights \(R^\ell\) proves the first two inequalities of \eqref{eq:local-polynomial-bounds}. This sum includes all zero-hit nonclean paths.

To reach port \(r\), expand over the preceding \(r-1\) slots. If none is selected, the active projection at the destination is itself the first selected port. Condition (i) then contributes \(a\sqrt{w_j}\) in this case as well; all other factors are contractions. Summing the terms bounds the query-input polynomial by \(a\sqrt{w_j}\rho^{r-1}\). The direct sum contains at most \(m\) ports, which gives the third inequality.

The clean and nonclean blocks have orthogonal ranges. For \(v=a^2w_j\rho^{2m}\), their squared bounds consequently give
\begin{align*}
 \norm{T_j(z)}^2
 &\le\bigl[1+a^2w_j(\rho^m-1)\bigr]^2
       +a^2w_j\rho^{2m}\\
 &\le1+3v+v^2\le e^{3v}.
\end{align*}
Taking square roots proves \eqref{eq:cell-growth}. All estimates are operator-norm estimates with arbitrary old registers, so they apply to the intermediate states produced by earlier deformed cells.
\end{proof}

Now let \(\Gamma(z)\) be the full query history of the deformed program. First-hit expansion, or the cut-open graph relation with each active oracle multiplied by \(z\), gives
\begin{equation}
 \Gamma(z)P_{\rm in}
 =(I-zA)^{-1}BP_{\rm in}
 =\sum_{r=0}^{M-1}z^rA^rBP_{\rm in}.
 \label{eq:history-generating}
\end{equation}
The inverse is a finite polynomial because \(A^M=0\). In particular \(\Gamma(1)P_{\rm in}=\Gamma P_{\rm in}\). There is no assertion that \(\Gamma(z)\) represents a normalized state when \(|z|>1\).

The component ending at a port in cell \(j\) is
\[
 K_j(z)T_{j-1}(z)\cdots T_1(z)P_{\rm in}.
\]
Oracle-independent isometries between cells can be inserted into this product without changing the bound. Lemma~\ref{lem:local-polynomial} and the orthogonal direct sum over final ports imply
\begin{equation}
\begin{aligned}
 \norm{\Gamma(z)P_{\rm in}}
 &\le a\sqrt{m\Omega}\,(1+R)^{m-1}\\
 &\quad\times\exp\!\left(\tfrac32a^2\Omega(1+R)^{2m}\right).
 \end{aligned}
 \label{eq:history-circle-bound}
\end{equation}
All interference within a final port is already included in its polynomial. Thus this estimate also covers histories with the same final environmental record and histories carrying records made before any specified private transition.

\begin{proof}[Proof of Theorem~\ref{thm:causal-tail}]
For \(R\ge2\), absorb powers of two in \eqref{eq:history-circle-bound} into constants \(B\ge1\) and \(D\ge1\), depending only on \(a,m\), to obtain
\begin{equation}
 \sup_{|z|=R}\norm{\Gamma(z)P_{\rm in}}
 \le B\sqrt\Omega\,R^{m-1}e^{D\Omega R^{2m}}.
 \label{eq:history-simple-bound}
\end{equation}
Applying Cauchy's integral formula to the operator polynomial in \eqref{eq:history-generating} gives
\[
 \norm{A^rBP_{\rm in}}
 \le B\sqrt\Omega\,R^{m-1-r}e^{D\Omega R^{2m}}.
\]
Since \(A^K\Gamma P_{\rm in}=\sum_{r=K}^{M-1}A^rBP_{\rm in}\), summing the geometric bound yields
\begin{equation}
 \norm{A^K\Gamma P_{\rm in}}
 \le \frac{B\sqrt\Omega\,R^{m-1}}{1-R^{-1}}
      e^{D\Omega R^{2m}}R^{-K}.
 \label{eq:cauchy-tail}
\end{equation}
When \(\Omega=0\), every query-input map vanishes and the result is immediate. Otherwise set
\[
 K=mq,
 \qquad R=\left(\frac{q}{4D\Omega}\right)^{1/(2m)}.
\]
For \(q\ge4D\,2^{2m}\Omega\), this radius is at least two. Moreover
\(\sqrt\Omega R^{m-1}=\sqrt{q/(4D)}/R\le\sqrt q\). Equation \eqref{eq:cauchy-tail} becomes
\begin{equation}
 \norm{A^{mq}\Gamma P_{\rm in}}
 \le2B\sqrt q\,e^{q/4}
       \left(\frac{4D\Omega}{q}\right)^{q/2}.
 \label{eq:tail-before-constants}
\end{equation}
For integer \(q\ge1\), \(\sqrt q\le2^{q/2}\). Absorb this into the base and prefactor in \eqref{eq:tail-main}, and increase \(c\) once more to cover the lower bound on \(q\).
\end{proof}

The power \(q/2\) comes from squaring the nonclean amplitude, which is of order \(\sqrt{w_j}\); the clean amplitude is already of order \(w_j\). We use orthogonality between these two outputs inside each cell. No orthogonality of the internal query histories is required.

\subsection{Choosing the query count}

\begin{lemma}[Inverting the tail]
\label{lem:tail-inversion}
Fix \(b_0,b_1\ge1\) and \(c_0>0\). There is a constant \(C\), depending only on these parameters, such that for \(\Omega>0\), \(0<\delta<1/2\), and \(L=\log(b_0/\delta)\), the integer
\begin{equation}
 q=\left\lceil C\left(1+\Omega+
 \frac{L}{\log(e+L/\Omega)}\right)\right\rceil
 \label{eq:tail-inversion-q}
\end{equation}
satisfies \(q\ge c_0\max\{\Omega,1\}\) and
\begin{equation}
 b_0\left(\frac{b_1\Omega}{q}\right)^{q/2}\le\delta.
 \label{eq:tail-inversion-goal}
\end{equation}
\end{lemma}

\begin{proof}
It suffices to arrange \(q\log(q/(b_1\Omega))\ge2L\). Put \(x=L/\Omega\). If \(x\le1\), the bound \(q\ge C\Omega\) gives this inequality once \(C\log(C/b_1)\ge2\). If \(x>1\), then
\[
 \frac q\Omega\ge\frac{Cx}{\log(e+x)}.
\]
For sufficiently large \(C\), uniformly on \(x\ge1\),
\[
 \log\!\left(\frac{Cx}{b_1\log(e+x)}\right)
 \ge\tfrac14\log(e+x).
\]
To see the uniformity, the difference of the two sides tends to positive infinity as \(x\to\infty\), and increasing \(C\) handles the remaining compact interval. Thus \(q\log(q/(b_1\Omega))\ge CL/4\), which is at least \(2L\) for \(C\ge8\). A larger constant also enforces the stated lower bound on \(q\).
\end{proof}

\begin{corollary}[Causal query compression]
\label{thm:compiler}
A program of \((a,m)\)-causal cells can be approximated on \(P_{\rm in}\) to isometry error \(0<\eps<1/8\) using
\begin{equation}
 O\!\left(m\left[1+\Omega+
 \frac{\log(1/\eps)}{\log(e+\log(1/\eps)/\Omega)}\right]\right)
 \label{eq:compiler-bound}
\end{equation}
controlled master queries and their adjoints. The constants depend only on \(a,m\). If \(\Omega=0\), no oracle query is needed.
\end{corollary}

\begin{proof}
For positive \(\Omega\), apply Lemma~\ref{lem:tail-inversion} to Theorem~\ref{thm:causal-tail} with target \(\eta_{mq}\le\eps/4\), then use Theorem~\ref{thm:tail-compiler} with \(K=mq\). The count is at most \(9mq\). If \(\Omega=0\), (ii) implies \(\Gamma P_{\rm in}=0\), hence \(DP_{\rm in}=VP_{\rm in}\) without any oracle calls.
\end{proof}

An isometry error \(\eta\) gives diamond error at most \(2\eta\). Indeed, pure-state trace distance gives this bound with a reference attached, and discarding the environment can only reduce it.

\section{A short Lindblad step from block encodings}
\label{sec:slice}

Consider a finite-dimensional Lindblad generator
\begin{equation}
 \cL(\rho)=-i[H,\rho]+\sum_{\mu=1}^r
 \left(L_\mu\rho L_\mu^\dagger-\tfrac12\{L_\mu^\dagger L_\mu,\rho\}\right).
 \label{eq:lindblad}
\end{equation}
In this section we are given exact blocks \(H/\alpha_H\) and \(L_\mu/\alpha_\mu\), with controlled queries and adjoints available. The jump label \(\mu\) may be selected coherently through SELECT. Ignore terms with zero normalization, and define
\begin{equation}
 \alpha_L^2=\sum_\mu\alpha_\mu^2,
 \qquad \Lambda=\alpha_H+\alpha_L^2.
 \label{eq:Lambda}
\end{equation}
We will construct a short-step isometry with \(O(h^2\Lambda^2)\) channel error and verify the weak-cell bounds with weight \(h\Lambda\) at its primitive query ports.

\subsection{The Euler column and its polar isometry}

Assume \(\alpha_L>0\) and prepare
\(\ket\chi=\alpha_L^{-1}\sum_\mu\alpha_\mu\ket\mu\) in the jump-label register.
Apply jump-SELECT once and project only the signal ancillas onto the
block-encoding subspace. Retaining the jump label gives the normalized
column operator
\begin{equation}
 L:\ket\psi\longmapsto\sum_\mu\ket\mu L_\mu\ket\psi,
 \qquad \frac{L}{\alpha_L}.
 \label{eq:stacked-L}
\end{equation}
We obtain $L^\dagger/\alpha_L$ from the adjoint at the same one-query cost. For $L^\dagger L/\alpha_L^2$, use an intermediate signal register to multiply the projected blocks; otherwise the unitaries simply cancel each other.

Set
\begin{equation}
 K=-iH-\tfrac12L^\dagger L,
 \qquad \mathcal A_h=\begin{pmatrix}I+hK\\ \sqrt h\,L\end{pmatrix}.
 \label{eq:Ah}
\end{equation}
Since \(K+K^\dagger=-L^\dagger L\), the terms linear in the step size cancel:
\begin{equation}
 \mathcal A_h^\dagger\mathcal A_h=I+h^2K^\dagger K.
 \label{eq:Ah-isometry}
\end{equation}
The polar factor is consequently
\begin{equation}
 V_h=\mathcal A_h(I+h^2K^\dagger K)^{-1/2}
 \label{eq:polar}
\end{equation}
and is isometric on the whole input space. Also,
\(\norm K\le\alpha_H+\alpha_L^2/2\le\Lambda\) and
\begin{equation}
 \diamondnorm{\cL}\le2\norm H+2\sum_\mu\norm{L_\mu}^2\le2\Lambda.
 \label{eq:L-diamond-bound}
\end{equation}

Trace out the row label of \(\mathcal A_h\) and call the resulting completely positive map \(\mathfrak A_h\). Expanding its two row contributions gives
\begin{equation}
\begin{gathered}
 \mathfrak A_h=\operatorname{Id}+h\cL+h^2\mathcal K,\\
 \mathcal K(\rho)=K\rho K^\dagger,\qquad
 \diamondnorm{\mathcal K}\le\Lambda^2.
 \end{gathered}
 \label{eq:A-channel-exact}
\end{equation}
Because \(e^{u\cL}\) is completely positive and trace preserving (CPTP), its diamond norm is one. The integral remainder therefore satisfies
\begin{align}
 e^{h\cL}-\operatorname{Id}-h\cL
 &=\int_0^h(h-u)\cL^2e^{u\cL}\,du,\notag\\
 \diamondnorm{e^{h\cL}-\operatorname{Id}-h\cL}
 &\le\tfrac12h^2\diamondnorm{\cL}^{\,2}\le2h^2\Lambda^2.
 \label{eq:exp-remainder}
\end{align}
In particular \(\diamondnorm{\mathfrak A_h-e^{h\cL}}\le3h^2\Lambda^2\).

For \(h\Lambda\le1\), the singular values of \(\mathcal A_h\) obey
\begin{equation}
 1\le\sigma\le\sqrt2,
 \qquad 0\le\sigma-1\le\tfrac12h^2\Lambda^2.
 \label{eq:Ah-singular-range}
\end{equation}
It follows that \(\norm{V_h-\mathcal A_h}\le h^2\Lambda^2/2\). If \(X,Y\) are two Stinespring operators with a common environment, their induced completely positive maps satisfy
\begin{equation}
 \diamondnorm{\mathcal X-\mathcal Y}
 \le(\norm X+\norm Y)\norm{X-Y}.
 \label{eq:stinespring-cp-bound}
\end{equation}
Indeed, expand their difference as
\((X-Y)\rho X^\dagger+Y\rho(X-Y)^\dagger\), apply the trace-norm product inequality with a reference system, and trace out the environment. Since \(\norm{\mathcal A_h}\le\sqrt2\), the channel \(\mathcal V_h\) induced by \(V_h\) satisfies
\begin{equation}
 \diamondnorm{\mathcal V_h-e^{h\cL}}
 \le\left(3+\frac{1+\sqrt2}{2}\right)h^2\Lambda^2
 \le5h^2\Lambda^2.
 \label{eq:polar-channel-error}
\end{equation}

\subsection{Encoding the column}

Let \(U_{L^\dagger L}\) denote the product block encoding above. A linear combination with coefficients \(1\), \(h\alpha_H\), and \(h\alpha_L^2/2\) encodes \(B_h=I+hK\). Explicitly, prepare
\begin{equation}
\begin{gathered}
 \ket{g_B}=\frac{\ket I+\sqrt{h\alpha_H}\ket H
                  +\sqrt{h\alpha_L^2/2}\ket D}{\sqrt{s_B}},\\
 s_B=1+h\alpha_H+\tfrac12h\alpha_L^2.
 \end{gathered}
 \label{eq:prepare-B}
\end{equation}
select \(I,-iU_H,-U_{L^\dagger L}\) on the three roles, and undo the role preparation. Projection gives \(B_h/s_B\). In the dissipative role, the two jump-SELECT calls have opposite directions. Identity padding fills these slots in the other roles.

Next prepare a stack qubit in
\begin{equation}
\begin{gathered}
 \ket{g_A}=\frac{s_B\ket0+\sqrt h\,\alpha_L\ket1}{\sqrt c},\\
 c=s_B^2+h\alpha_L^2,\qquad s=c^{-1/2}.
 \end{gathered}
 \label{eq:prepare-stack}
\end{equation}
On its zero branch apply the encoding of \(B_h/s_B\), and on its one branch apply the rectangular encoding of \(L/\alpha_L\). The stack qubit is retained in the output. With input and output signal projectors \(\Pi_{\rm in}\) and \(\Pi_{\rm out}\), this defines a unitary \(W_h\) such that
\begin{equation}
 \Pi_{\rm out}W_h\Pi_{\rm in}=s\mathcal A_h.
 \label{eq:explicit-Wh}
\end{equation}
All signal projectors refer to known ancilla states or subspaces.

Complete each coefficient preparation to a unitary by rotating in the plane spanned by the reference and prepared states, through the minimal angle, and fixing the orthogonal complement. On padding branches, fix all ancillary registers. This specifies the whole unitary, so the comparison with an oracle-bypassing circuit also applies when we run the preparation in reverse.

The normalization satisfies
\begin{equation}
\begin{gathered}
 c=1+2h\Lambda+h^2(\alpha_H+\alpha_L^2/2)^2,\\
 1+h^2\norm K^2\le c\le(1+h\Lambda)^2.
 \end{gathered}
 \label{eq:c-bounds}
\end{equation}
Hence \(s\norm{\mathcal A_h}\le1\), and \(s\in[1/2,1]\) for \(h\Lambda\le1\).

\subsection{An explicit three-call correction}

The polar correction can be implemented by a short phase-modulated amplitude-amplification sequence. It is a cubic instance of projected-unitary singular value transformation~\cite{GSLW19}, but its phases and full output error can be derived directly.

\begin{lemma}[Cubic polar correction]
\label{lem:cubic-polar}
Suppose \(W\) is a unitary with projected block \(\Pi_LW\Pi_R=sA\), where \(s\in[1/2,1]\), \(s\norm A\le1\), and every singular value of \(A\) lies in \([1,\sqrt2]\). Let \(A=V|A|\), and define
\begin{equation}
\begin{gathered}
 d_s=\sqrt{4s^2-1},\qquad g_s=-\frac{1+id_s}{2s},\\
 u_s=1-\frac1{2s^2}+i\frac{d_s}{2s^2}.
 \end{gathered}
 \label{eq:explicit-phases}
\end{equation}
Both \(u_s\) and \(g_s\) have modulus one. For \(\Phi_P(u)=I+(u-1)P\), the unitary
\begin{equation}
 U_s=g_sW\Phi_{\Pi_R}(u_s)W^\dagger\Phi_{\Pi_L}(u_s)W
 \label{eq:explicit-cubic-sequence}
\end{equation}
uses three signal calls. Its restriction \(J=U_s\Pi_R\) satisfies
\begin{equation}
 \norm{J-V}^2\le\max_\sigma(\sigma-1)^2(\sigma+2),
 \label{eq:cubic-full-column}
\end{equation}
where the target \(V\) is embedded in the output signal subspace.
\end{lemma}

\begin{proof}
Write \(X=\Pi_LW\Pi_R\). Multiplying the two phase factors gives
\begin{equation}
 \Pi_LU_s\Pi_R
 =g_s\bigl[(2u_s-1)X+(u_s-1)^2XX^\dagger X\bigr].
 \label{eq:cubic-block-multiplication}
\end{equation}
The scalar identities
\[
 g_s(2u_s-1)=\frac{3-id_s}{2s},
 \qquad g_s(u_s-1)^2=\frac{-1+id_s}{2s^3}
\]
show that the projected output on a right singular vector \(v_\sigma\), with polar left vector \(Vv_\sigma\), has amplitude
\begin{equation}
 F_s(\sigma)=\frac{3\sigma-\sigma^3}{2}
       +i\frac{d_s}{2}(\sigma^3-\sigma).
 \label{eq:cubic-singular-amplitude}
\end{equation}
In particular its real part is independent of \(s\). Since \(J\) and \(V\) are isometries on the input space,
\begin{align}
 (J-V)^\dagger(J-V)
 &=2I-V^\dagger J-J^\dagger V\notag\\
 &=2I-3|A|+|A|^3\notag\\
 &= (|A|-I)^2(|A|+2I).
 \label{eq:cubic-error-operator}
\end{align}
Taking the largest eigenvalue proves \eqref{eq:cubic-full-column}. This also controls all output components outside the signal subspace.
\end{proof}

At \(s=1/2\), the phases are \(u_s=g_s=-1\), recovering the usual cubic amplification sequence. At every allowed \(s\), \(F_s(1)=1\). The phases depend only on the known normalization, and their implementation introduces no oracle queries.

Applying the lemma to \(W_h\) and \(\mathcal A_h\), equation \eqref{eq:Ah-singular-range} gives
\begin{equation}
 \norm{J_h-V_h}\le h^2\Lambda^2,
 \qquad
 \diamondnorm{\mathcal J_h-e^{h\cL}}\le7h^2\Lambda^2.
 \label{eq:local-error}
\end{equation}
The channel estimate follows from \eqref{eq:polar-channel-error} and the factor-two conversion between isometry and channel error.

\subsection{Bounds at the primitive oracle ports}

The channel estimate is not sufficient for compression: the weak-cell definition concerns individual query ports. In particular we must include the ports in \(W_h^\dagger\). We therefore treat the three signal calls as one cell, supplying a fresh bath and retaining the block-encoding workspace at the output.

\begin{lemma}[Weakness of the short-step circuit]
\label{lem:slice-skeleton}
The circuit \(J_h\), with the preparation extensions specified above, is an \((a,m)\)-causal cell of weight \(h\Lambda\), for universal constants \(a\) and \(m\).
\end{lemma}

\begin{proof}
Flatten \eqref{eq:explicit-cubic-sequence} into its \(m=O(1)\) primitive oracle slots. Compare its inter-slot unitaries \(F_r(h)\) with unitaries \(F_r^0\) obtained by replacing the preparations of \(g_A\) and \(g_B\), and their inverses, by their zero-coupling versions. Keep the actual phases \(u_s,g_s\), the signal projectors, and the slot order fixed during this comparison.

The probabilities outside the reference components are
\begin{equation}
 1-\frac1{s_B}=\frac{h\alpha_H+h\alpha_L^2/2}{s_B}\le h\Lambda,
 \qquad \frac{h\alpha_L^2}{c}\le h\Lambda.
 \label{eq:off-probabilities}
\end{equation}
For a minimal rotation preparing a state with reference overlap \(\cos\theta\ge0\), the distance from the identity is \(2\sin(\theta/2)\), at most twice the norm of the off-reference component. Each inter-slot unitary contains only a constant number of such rotations. Therefore
\begin{equation}
 \max_r\norm{F_r(h)-F_r^0}\le c\sqrt{h\Lambda}
 \label{eq:free-map-hybrid}
\end{equation}
for a universal constant \(c\).

The stack qubit is on the no-jump branch of the comparison circuit, while each role register selects identity and each block ancilla stays in its reference state. These labels are preserved by the signal phase rotations. Hence the oracle slots act trivially on the initialized trajectory. The comparison circuit may acquire a known phase, but it ends in the same clean ancillary subspace. No removal of this phase is needed.

Let \(\mathcal K_r\) be its clean trajectory subspace immediately before slot \(r\), with the system and old registers unrestricted. Then
\begin{equation}
 P_r\mathcal K_r=0,
 \qquad Q_r\mathcal K_r=\mathcal K_r,
 \qquad F_r^0\mathcal K_r=\mathcal K_{r+1}.
 \label{eq:skeleton-invariant}
\end{equation}
The final subspace \(\mathcal K_{\rm out}\) is the clean boundary subspace. This statement also holds when the phases are controlled by an old history register: the system and history are unrestricted in \(\mathcal K_r\), and the reference labels remain fixed.

For products of contractions, the telescoping identity gives
\begin{equation}
 \norm{G_b\cdots G_a-G_b^0\cdots G_a^0}
 \le\sum_{r=a}^b\norm{G_r-G_r^0}.
 \label{eq:finite-hybrid}
\end{equation}
There are only a constant number of factors. Along an absorbed prefix the comparison trajectory is annihilated by the final active projection, so \eqref{eq:free-map-hybrid} and \eqref{eq:finite-hybrid} give condition (i). The same comparison with the ordinary, unabsorbed prefix gives an \(O(\sqrt{h\Lambda})\) bound at each query-input port; taking their finite direct sum gives condition (ii).

For (iii), begin with an active state in \(\mathcal K_r^\perp\). Unitarity carries orthogonal complements of the clean subspaces to one another in the comparison circuit. Each inactive projector also preserves the appropriate complement, since it is self-adjoint and is the identity on the clean subspace. No clean component can therefore reach the output of the comparison suffix. Its \(O(\sqrt{h\Lambda})\) distance from the actual suffix proves (iii). For (iv) we absorb all slots. The comparison ends in \(\mathcal K_{\rm out}\) and has no nonclean component; \eqref{eq:finite-hybrid} then bounds the actual nonclean output by \(O(\sqrt{h\Lambda})\).

The three signal calls constitute one cell. In particular, cancellation of a component of \(W_h\) by \(W_h^\dagger\) is already covered by the cell estimates. They are operator-norm estimates, valid with a reference system. Keep the fresh workspace after the cell and let every later cell act trivially on it.
\end{proof}

\begin{proposition}[Lindblad cell]
\label{prop:lindblad-cell}
For \(h\Lambda\le1\), the isometry \(J_h\) is implementable with a constant number of controlled Hamiltonian and jump-SELECT queries and their adjoints. It is an \((a,m)\)-causal cell of weight \(h\Lambda\), and its reduced channel has diamond error at most \(7h^2\Lambda^2\) from \(e^{h\cL}\).
\end{proposition}

For a time-dependent generator, the mesh must be fine enough to control the error from freezing the generator as well as the accumulated channel error. Once that mesh is fixed, we apply the compiler with the sum of the cell weights.

\section{Time-dependent Lindblad evolution}
\label{sec:time-dependent}

Let
\begin{equation}
\begin{aligned}
 \cL_t(\rho)={}&-i[H(t),\rho]\\
 &+\sum_{\mu=1}^r\Bigl(L_\mu(t)\rho L_\mu(t)^\dagger\\
 &\qquad-\tfrac12\{L_\mu(t)^\dagger L_\mu(t),\rho\}\Bigr)
 \end{aligned}
 \label{eq:time-lindblad}
\end{equation}
act on a finite-dimensional system for \(t\in[0,T]\). Choose known normalizations valid throughout this interval and write
\begin{equation}
 \Lambda=\alpha_H+\sum_\mu\alpha_\mu^2,
 \qquad \tau=T\Lambda.
 \label{eq:tau}
\end{equation}
For a grid \(t_j=jT/J\), the access model supplies exact encodings
\begin{equation}
\begin{aligned}
 (\bra0\otimes I)U_H(t_j)(\ket0\otimes I)&=\frac{H(t_j)}{\alpha_H},\\
 (\bra0\otimes I)U_\mu(t_j)(\ket0\otimes I)&=\frac{L_\mu(t_j)}{\alpha_\mu},
 \end{aligned}
 \label{eq:time-blocks}
\end{equation}
with common ancillary spaces after padding. We assume coherent selection over the labels:
\begin{equation}
 \sum_j\ketbra j j\otimes U_H(t_j),
 \qquad
 \sum_{j,\mu}\ketbra{j,\mu}{j,\mu}\otimes U_\mu(t_j),
 \label{eq:time-select}
\end{equation}
and controlled access to these unitaries and their adjoints, for the grid chosen by the algorithm.

There are a constant number of oracle roles within a slice: Hamiltonian, jump-SELECT, the factors in the encoding of \(L^\dagger L\), and identity padding. With \(U^{(+)}=U\), \(U^{(-)}=U^\dagger\), they can be collected into
\begin{equation}
 O_J^{\rm mast}
 =\sum_{j,\varrho,d}\ketbra{j,\varrho,d}{j,\varrho,d}
       \otimes U_\varrho(t_j)^{(d)}.
 \label{eq:master-oracle}
\end{equation}
The private port specifies which time, role, direction, and padding registers a call to \(O_J^{\rm mast}\) needs. Reversibly compute those labels, make the master call, and undo the computation. This implements the private oracle of \eqref{eq:cut-open-transducer}; jump-SELECT handles the jump label. It uses one master query, hence only a constant number of the supplied coherent queries.

The normalization constants are known, so no queries enter the coefficient states. We assume exact coherent selection, control, and adjoints throughout the chosen grid. At each grid interval, apply Section~\ref{sec:slice} to the supplied encodings of \(H(t)\) and \(L_\mu(t)\).

\begin{theorem}[Time-dependent Lindblad simulation]
\label{thm:lindblad}
Assume \eqref{eq:time-blocks}--\eqref{eq:master-oracle} and a finite constant \(\beta\) such that
\begin{equation}
 \diamondnorm{\cL_t-\cL_s}\le\beta|t-s|.
 \label{eq:lipschitz}
\end{equation}
For \(0<\eps<1/8\), the propagator
\begin{equation}
 \Phi(T,0)=\mathcal T\exp\!\left(\int_0^T\cL_t\,dt\right)
 \label{eq:target-propagator}
\end{equation}
can be approximated in diamond norm to error \(\eps\) with
\begin{equation}
 O\!\left(1+\tau+
 \frac{\log(1/\eps)}{\log(e+\log(1/\eps)/\tau)}\right)
 \label{eq:lindblad-query-bound}
\end{equation}
controlled master queries and their adjoints. At \(\tau=0\), the simulation makes no queries. At \(\tau\ge1\), the worst-case query count in this access model also has a lower bound of the stated order. This remains true when all the jump operators vanish.
\end{theorem}

\begin{proof}
Assume \(\tau>0\). Choose an integer
\begin{equation}
\begin{gathered}
 J\ge\max\left\{1,\lceil\tau\rceil,
       \left\lceil\frac{2\beta T^2}{\eps}\right\rceil,
       \left\lceil\frac{28\tau^2}{\eps}\right\rceil\right\},\\
 h=T/J.
 \end{gathered}
 \label{eq:grid-choice}
\end{equation}
Define the frozen-generator channel and the product of short-step channels by
\begin{equation}
\begin{aligned}
 \Phi_{\rm grid}&=e^{h\cL_{t_{J-1}}}\cdots e^{h\cL_{t_0}},\\
 \Phi_{\rm cells}&=\mathcal J_{h,J-1}\cdots\mathcal J_{h,0}.
 \end{aligned}
 \label{eq:grid-products}
\end{equation}
Apply variation of constants with the piecewise constant generator \(\cL_{t_j}\). CPTP contractivity removes the surrounding propagators from the diamond-norm estimate, leaving
\begin{align}
 \diamondnorm{\Phi(T,0)-\Phi_{\rm grid}}
 &\le\sum_{j=0}^{J-1}\int_{t_j}^{t_{j+1}}
       \diamondnorm{\cL_u-\cL_{t_j}}\,du\notag\\
 &\le\frac{\beta T^2}{2J}\le\frac\eps4.
 \label{eq:freezing-error}
\end{align}
With \(h\Lambda\le1\), Proposition~\ref{prop:lindblad-cell} applies slice by slice. A telescoping sum of channel differences yields
\begin{equation}
 \diamondnorm{\Phi_{\rm grid}-\Phi_{\rm cells}}
 \le7Jh^2\Lambda^2=\frac{7\tau^2}{J}\le\frac\eps4.
 \label{eq:cell-product-error}
\end{equation}

Let \(V_{\rm cells}\) be the product Stinespring isometry with all slice environments retained. Each cell has weight \(h\Lambda\), hence
\begin{equation}
 \Omega=\sum_{j=0}^{J-1}h\Lambda=\tau.
 \label{eq:weight-sum}
\end{equation}
Apply Corollary~\ref{thm:compiler} with isometry error \(\eps/4\). The compressed isometry induces a channel \(\widetilde\Phi\) with
\begin{equation}
 \diamondnorm{\widetilde\Phi-\Phi_{\rm cells}}
 \le\frac\eps2.
 \label{eq:compression-channel-error}
\end{equation}
Summing \eqref{eq:freezing-error}, \eqref{eq:cell-product-error}, and \eqref{eq:compression-channel-error} proves the error bound. In the query estimate, insert the constant \(m\) and total weight \(\Omega=\tau\) to obtain \eqref{eq:lindblad-query-bound}. There is no \(J\) dependence.

For a lower bound it suffices to set the jumps to zero and make the Hamiltonian constant in time. We use a control qubit to keep its phase observable at the channel level:
\begin{equation}
 \overline H=\ketbra0 0\otimes0+\ketbra1 1\otimes H.
 \label{eq:controlled-H}
\end{equation}
For \(\overline H\), control the encoding of \(H\) and use a known ancillary operation to give a zero block on the other branch. This costs one controlled call. Its evolution is \(I\oplus e^{-iHT}\), so the input \(\ket+\ket\psi\), with system state \(\ket\psi\), becomes
\[
 \frac{\ket0\ket\psi+\ket1e^{-iHT}\ket\psi}{\sqrt2}.
\]
A phase multiplying \(e^{-iHT}\) changes the relative phase of these two branches. It is therefore visible to a channel simulator. The controlled Hamiltonian lower bound~\cite{LowChuang17,CGWZ26} retains the phase information it needs. Since the construction is a valid zero-jump Lindblad instance, that lower bound matches \eqref{eq:lindblad-query-bound} for \(\tau\ge1\) in the worst case.
\end{proof}

For very small action, the upper bound need not be interpreted as a positive lower bound. Indeed, another variation-of-constants estimate gives \(\diamondnorm{\Phi(T,0)-\operatorname{Id}}\le2\tau\); if \(2\tau\le\eps\), the identity channel already meets the target accuracy.

For gate costs without additional input structure, let \(G_{\rm route}(J)\) be the cost of visiting the \(J\) cells once without the oracle. There are \(O(q)\) such visits, costing \(O(qG_{\rm route}(J))\) gates, before accounting for reuse and selection. The output also stores all \(J\) slice workspaces. Nothing in the query bound controls these resources without further information about the preparation and routing circuits. If the given encodings can only be indexed classically, even a coherent SELECT may cost order \(J\) queries.

For efficiently evaluable local matrices, Sections~\ref{loc:sec:cells}--\ref{loc:sec:assembly} provide the additional structure needed to implement the routing and retained baths with nearly linear total gate cost.

\section{Adaptive Markovian query computation}
\label{sec:adaptive}

An oracle algorithm may use dissipation as well as Hamiltonian evolution, with measurement outcomes determining its later couplings. We measure its cost by oracle action: the integrated strength of the interactions specified below. Under the access and regularity conditions of this section, minimizing that action gives ordinary quantum query complexity up to constant factors. For this comparison, take a standard Hermitian involutory oracle
\begin{equation}
 O_x=O_x^\dagger=O_x^{-1},
 \label{eq:standard-oracle}
\end{equation}
with a known identity sector for coherent bypass. Controlled uses have the usual constant query overhead. The description of a protocol, including its controls, supplied normalizations, and initial state preparation, is common to all oracle inputs \(x\).

\begin{definition}[Adaptive Markovian query protocol]
\label{def:adaptive}
A protocol has finite-dimensional accessible workspace, finitely many continuous intervals, and finitely many discrete interventions. On a continuous interval its generator is the sum of an oracle-independent GKSL generator \(F_t\) and an oracle-dependent term selected by the classical history \(r\):
\begin{equation}
\begin{aligned}
 Q_{x,r,t}(\rho)&=-i[H_{x,r}(t),\rho]\\
 &\quad+\sum_\mu\mathcal D[L_{x,r,\mu}(t)](\rho),\\
 \mathcal D[L](\rho)&=L\rho L^\dagger-\tfrac12\{L^\dagger L,\rho\}.
 \end{aligned}
 \label{eq:adaptive-generator}
\end{equation}
The interventions are oracle-independent CPTP maps or instruments with finite outcome sets. Their outcomes may select subsequent controls and generators. The displayed operators and generators are uniformly bounded, and all jump-label sets are finite.

For a fixed \(b\ge1\), the block encodings of \(H_{x,r}(t)\) and \(L_{x,r,\mu}(t)\) can each be implemented with at most \(b\) calls to \(O_x\). We require this bound also for their coherent SELECT: time, history, jump label, role, and the choice between an encoding and its adjoint may all be in superposition. Known oracle-independent normalizations satisfy
\begin{equation}
 \norm{H_{x,r}(t)}\le\alpha_{H,r}(t),
 \qquad \norm{L_{x,r,\mu}(t)}\le\alpha_{r,\mu}(t).
 \label{eq:adaptive-normalizations}
\end{equation}
The action density and total oracle action are
\begin{equation}
\begin{gathered}
 \lambda(t)=\sup_r\left(\alpha_{H,r}(t)
                     +\sum_\mu\alpha_{r,\mu}(t)^2\right),\\
 \tau=\int_0^T\lambda(t)\,dt.
 \end{gathered}
 \label{eq:adaptive-action}
\end{equation}
We require \(\lambda\) to be bounded and Riemann integrable on each interval. The chronology is fixed after padding branches, discarded environments are fresh and never reused, and the protocol is trace preserving.
\end{definition}

Padding puts all branches on a common schedule and assigns zero oracle generator after a branch halts. Keep an append-only register of outcomes and use it to control later operations. If the protocol needs to overwrite its history, it can do so on a work copy while the recorded outcomes stay intact. In \eqref{eq:adaptive-action}, the supremum includes every history, regardless of its probability. Thus the action is a uniform bound over branches. An oracle-independent term inside a jump operator is included in that operator's normalization; only an additive oracle-independent GKSL generator in \(F_t\) is free.

\subsection{Retaining classical histories coherently}

A measurement record can be retained in an isometric implementation, but the extension of a history-conditioned generator to coherent records must still be specified. Different choices can agree on classical histories and act differently on their off-diagonal blocks. The following extension provides a convenient common implementation.

\begin{lemma}[History-controlled extension]
\label{lem:history-extension}
Let \(P_r=\ketbra r r_R\). On the history and system registers define
\begin{equation}
\begin{aligned}
 \widehat H_x(t)&=\sum_rP_r\otimes H_{x,r}(t),\\
 \widehat L_{x,r,\mu}(t)&=P_r\otimes L_{x,r,\mu}(t),
 \end{aligned}
 \label{eq:history-extension}
\end{equation}
and let \(\widehat Q_{x,t}\) be their GKSL generator. Suppose an environment \(E\) retains the history in orthogonal subspaces, so that
\begin{equation}
\begin{gathered}
 \rho_{RSE}=\sum_{r,s}(P_r\otimes\Pi_r^E)\rho_{RSE}
                             (P_s\otimes\Pi_s^E),\\
 \Pi_r^E\Pi_s^E=\delta_{r,s}\Pi_r^E.
 \end{gathered}
 \label{eq:history-record-state}
\end{equation}
Then
\begin{equation}
 \Tr_E\bigl[(e^{h\widehat Q_{x,t}}\otimes\operatorname{Id}_E)(\rho_{RSE})\bigr]
 =\sum_rP_r\otimes e^{hQ_{x,r,t}}(\rho_{rr}),
 \label{eq:history-extension-channel}
\end{equation}
where \(\rho_{rr}\) is the subnormalized system state in the \(r\)-th diagonal block after tracing \(E\).

For \(h\lambda(t)\le1\), this extended evolution admits a causal short-step cell of weight \(h\lambda(t)\), with diamond error \(O(h^2\lambda(t)^2)\). A master call in the cell uses at most \(b\) ordinary queries.
\end{lemma}

\begin{proof}
Tracing out \(E\) removes the off-diagonal history blocks. The generator preserves the block-diagonal algebra and satisfies
\[
 \widehat Q_{x,t}\!\left(\sum_rP_r\otimes\rho_{rr}\right)
 =\sum_rP_r\otimes Q_{x,r,t}(\rho_{rr}).
\]
Exponentiation proves \eqref{eq:history-extension-channel}.

For the short-step implementation, use the rectangular jump map
\begin{equation}
 \widehat L_x(t):\ket r\ket\psi
 \longmapsto\sum_\mu\ket{r,\mu}_J\ket r_R
                            L_{x,r,\mu}(t)\ket\psi.
 \label{eq:history-stacked-jump}
\end{equation}
The jump environment explicitly retains \(r\); this implements the separate operators in \eqref{eq:history-extension}. In particular,
\begin{equation}
\begin{gathered}
 \widehat L_x(t)^\dagger\widehat L_x(t)
 =\sum_rP_r\otimes\sum_\mu L_{x,r,\mu}(t)^\dagger L_{x,r,\mu}(t),\\
 \norm{\widehat L_x(t)}^2\le\lambda(t).
 \end{gathered}
 \label{eq:history-rate-bound}
\end{equation}
If \(\widehat K=-i\widehat H-\widehat L^\dagger\widehat L/2\), then \(\norm{\widehat K}\le\lambda(t)\). Writing
\(\widehat Q(\rho)=\widehat K\rho+\rho\widehat K^\dagger+\Tr_J(\widehat L\rho\widehat L^\dagger)\) also gives \(\diamondnorm{\widehat Q}\le3\lambda(t)\).

Apply the column construction of Section~\ref{sec:slice} in each history block, with its own known coefficient preparations and phases. The no-jump output has a common environmental label, while jump outputs carry the label \((r,\mu)\) in \eqref{eq:history-stacked-jump}. The input and output history sectors remain orthogonal. Thus the full isometry bounds are the maxima of the branch bounds, and every primitive port has amplitude \(O(\sqrt{h\lambda(t)})\).

To check the channel error also on coherent histories, form the global Euler column from \(\widehat K\) and \(\widehat L\). It satisfies
\(\widehat{\mathcal A}_h^\dagger\widehat{\mathcal A}_h=I+h^2\widehat K^\dagger\widehat K\). Its induced completely positive map is exactly
\(\operatorname{Id}+h\widehat Q+h^2(\widehat K\,\cdot\,\widehat K^\dagger)\).
The remainder estimate \eqref{eq:exp-remainder}, now with \(\diamondnorm{\widehat Q}\le3\lambda(t)\), and the same polar and phase-correction estimates give a diamond error at most \(9h^2\lambda(t)^2\). The bound holds on the whole coherent extension, including its off-diagonal blocks. To implement a master slot, use the assumed history SELECT and pad with identities as needed. This takes at most \(b\) ordinary queries.
\end{proof}

\subsection{Uniform discretization}

We impose a common time-regularity condition on the padded protocol. On every continuous interval \(I\), let \(B_I\) be a finite set of breakpoints, independent of \(x\). Assume
\begin{equation}
\begin{aligned}
 M_F&:=\sup_{I,t\in I\setminus B_I}\diamondnorm{F_t}<\infty,\\
 M_Q&:=\sup_{I,x,t\in I\setminus B_I}\diamondnorm{\widehat Q_{x,t}}<\infty.
 \end{aligned}
 \label{eq:uniform-bounds}
\end{equation}
For each \(I\), let \(S_\eta(I)\) consist of the pairs \(s,t\) in the same component of \(I\setminus B_I\) with \(|t-s|\le\eta\). Assume that the common modulus
\begin{equation}
\begin{aligned}
 \omega(\eta):={}&\sup_{\substack{I,x\\(s,t)\in S_\eta(I)}}
 \Bigl(\diamondnorm{F_t-F_s}\\
 &\qquad+\diamondnorm{\widehat Q_{x,t}-\widehat Q_{x,s}}\Bigr)\\
 &\longrightarrow0\quad(\eta\downarrow0).
 \end{aligned}
 \label{eq:uniform-modulus}
\end{equation}
The free generator here is an oracle-independent extension to the history register that preserves its classical sectors. Generator values at the finitely many breakpoints do not affect the propagator.

A sufficient condition is piecewise operator-norm continuity of \(H_{x,r}(t)\) and \(L_{x,r,\mu}(t)\), with a common finite partition and moduli uniform in \(x,r,\mu\), together with the corresponding diamond-norm continuity of \(F_t\). The finite history and jump-label sets and the inequalities
\begin{equation}
\begin{gathered}
 \diamondnorm{-i[H-K,\,\cdot\,]}\le2\norm{H-K},\\
 \diamondnorm{\mathcal D[L]-\mathcal D[M]}
 \le2(\norm L+\norm M)\norm{L-M}
 \end{gathered}
 \label{eq:operator-continuity}
\end{equation}
then imply \eqref{eq:uniform-modulus}.

\begin{theorem}[Digitalization of adaptive Markovian queries]
\label{thm:adaptive}
Let \(\mathcal A_x\) be the channel of a protocol in Definition~\ref{def:adaptive} satisfying \eqref{eq:uniform-bounds}--\eqref{eq:uniform-modulus}, with action \(\tau\). For \(0<\delta<1/8\), an ordinary quantum query circuit implements a channel \(\widetilde{\mathcal A}_x\) such that
\begin{equation}
 \sup_x\diamondnorm{\widetilde{\mathcal A}_x-\mathcal A_x}\le\delta
 \label{eq:adaptive-error}
\end{equation}
using
\begin{equation}
 O\!\left(b\left[1+\tau+
 \frac{\log(1/\delta)}{\log(e+\log(1/\delta)/\tau)}\right]\right)
 \label{eq:adaptive-query-bound}
\end{equation}
queries to \(O_x\). The accessible classical record may be included in the output. At \(\tau=0\), no oracle query is required.
\end{theorem}

\begin{proof}
First retain the environments of all interventions. An instrument with Kraus operators \(K_{y,a}\) can be implemented by the isometry
\begin{equation}
 W_{\mathcal M}=\sum_{y,a}\ket y_R\ket{y,a}_E K_{y,a}.
 \label{eq:purified-instrument}
\end{equation}
The relation \(\sum_{y,a}K_{y,a}^\dagger K_{y,a}=I\) makes this an isometry. Use \(R\) as the control register for future steps and retain \(E\) without further action. Iterating gives \eqref{eq:history-record-state}. Dilate the other oracle-independent channels in the same manner. The feedback evolution now has the chronological workspace of Lemma~\ref{lem:history-extension}.

For a common discretization over \(x\), include every intervention and breakpoint in \(\Pi=\{I_j=[t_j,t_{j+1}]\}\). Let \(h_j=t_{j+1}-t_j\) and take \(s_j\) in the regularity piece containing the interior of \(I_j\). Freeze the continuous generator by replacing
\[
 G_x(t)=F_t+\widehat Q_{x,t}
 \quad\text{to}\quad
 G_{x,\Pi}(t)=F_{s_j}+\widehat Q_{x,s_j}\quad(t\in I_j).
\]
If \(U_x\) and \(U_{x,\Pi}\) are the corresponding propagators, variation of constants and CPTP contractivity give
\begin{equation}
\begin{aligned}
 &\sup_x\diamondnorm{U_x-U_{x,\Pi}}\\
 &\quad\le\sum_j\int_{I_j}\sup_x\diamondnorm{G_x(t)-G_x(s_j)}\,dt\\
 &\quad\le T\omega(\norm\Pi).
 \end{aligned}
 \label{eq:adaptive-freezing}
\end{equation}
The mesh width is \(\norm\Pi=\max_jh_j\). In the sum over intervals, \(T\) is their total continuous duration.

For the frozen generators, the Lie--Trotter error is
\begin{equation}
\begin{aligned}
 &\diamondnorm{e^{h_j(F_{s_j}+\widehat Q_{x,s_j})}
             -e^{h_jF_{s_j}}e^{h_j\widehat Q_{x,s_j}}}\\
 &\quad\le\tfrac12h_j^2\diamondnorm{[F_{s_j},\widehat Q_{x,s_j}]}\\
 &\quad\le h_j^2M_FM_Q.
 \end{aligned}
 \label{eq:adaptive-splitting}
\end{equation}
For completeness, differentiating the product
\[
 e^{(h-u)(F+Q)}e^{uF}e^{uQ}
\]
writes the difference as an integral containing \([e^{uF},Q]\). Express this commutator by integrating over \(v\in[0,u]\); the resulting double integral ranges over \(0\le v\le u\le h\). Its integrand is \([F,Q]\) multiplied on either side by CPTP exponentials of diamond norm one. The first inequality in \eqref{eq:adaptive-splitting} follows.

Telescoping across intervals and interventions consequently bounds the total freezing and splitting error by
\begin{equation}
 T\omega(\norm\Pi)+M_FM_QT\norm\Pi\longrightarrow0.
 \label{eq:uniform-product-limit}
\end{equation}
The common modulus and breakpoints are what make this convergence uniform in the oracle input.

Use the same partition to define the upper action weights
\begin{equation}
\begin{gathered}
 \lambda_j^+=\sup_{t\in I_j}\lambda(t),\qquad w_j=h_j\lambda_j^+,\\
 \Omega=\sum_jw_j=U(\lambda,\Pi).
 \end{gathered}
 \label{eq:darboux-weights}
\end{equation}
The free factors \(e^{h_jF_{s_j}}\) and all discrete interventions have oracle-independent Stinespring implementations. Replace each \(e^{h_j\widehat Q_{x,s_j}}\) by the cell in Lemma~\ref{lem:history-extension}. Its branch weights are at most \(w_j\), and its local channel error is at most \(9w_j^2\), once \(w_j\le1\).

Bounded Riemann integrability gives
\begin{equation}
\begin{gathered}
 \max_jw_j\le\norm\lambda_\infty\norm\Pi\longrightarrow0,
 \qquad \Omega\longrightarrow\tau,\\
 \sum_jw_j^2\le(\max_jw_j)\Omega\longrightarrow0.
 \end{gathered}
 \label{eq:darboux-convergence}
\end{equation}
For \(\tau>0\), refine the partition until \(\Omega\le2\tau\) and the sum of freezing, splitting, and cell errors is at most \(\delta/2\) for every \(x\).

The product is chronological and isometric, with free contractions on absorbed paths and no reuse of retained environments. We may therefore use Corollary~\ref{thm:compiler} at isometry accuracy \(\delta/4\). Let \(\mathcal A_{\Pi,x}\) denote the channel on this partition. The two approximations obey
\[
\begin{aligned}
 \sup_x\diamondnorm{\widetilde{\mathcal A}_x-\mathcal A_{\Pi,x}}&\le\delta/2,\\
 \sup_x\diamondnorm{\mathcal A_{\Pi,x}-\mathcal A_x}&\le\delta/2.
\end{aligned}
\]
Their sum bounds the total error in \eqref{eq:adaptive-error}. Substitution of \(\Omega\le2\tau\) in \eqref{eq:compiler-bound}, with at most \(b\) queries at a master slot, gives \eqref{eq:adaptive-query-bound}.

For \(\tau=0\) there is no oracle coupling except on a set of measure zero, so the free dynamics and interventions implement the protocol without queries.
\end{proof}

\subsection{An adversary characterization}

Let \(f:\mathcal D\to E\), with \(\mathcal D\subseteq\Sigma^n\) and \(\Sigma,E\) finite; the function may be partial or non-Boolean. Let \(Q(f)\) be its ordinary quantum query complexity with success probability at least \(2/3\). Fix the access constant \(b=O(1)\). Among protocols in Definition~\ref{def:adaptive} that \emph{satisfy \eqref{eq:uniform-bounds}--\eqref{eq:uniform-modulus}} and compute \(f\) with success probability at least \(2/3\), take the infimum of the action \(\tau\). We denote this infimum by \(MQ(f)\).

\begin{corollary}[Open-system adversary completeness]
\label{cor:adversary}
For every such function,
\begin{equation}
 MQ(f)=\Theta(Q(f))=\Theta(\operatorname{Adv}^{\pm}(f)).
 \label{eq:adversary-completeness}
\end{equation}
If \(f\) is constant, each quantity is zero. Otherwise the constants in the \(\Theta\)-relations can be chosen uniformly over \(f\), with dependence only on the fixed access constant \(b\).
\end{corollary}

\begin{proof}
Set \(\delta=1/16\) in Theorem~\ref{thm:adaptive}. The total-variation error is at most \(1/32\), leaving success probability at least \(2/3-1/32>1/2\). A constant number of repetitions restores success to \(2/3\), and hence
\begin{equation}
 Q(f)=O\bigl(b(1+\tau)\bigr).
 \label{eq:query-from-action}
\end{equation}
For a nonconstant function, the additive one can be removed by an action lower bound.

For any fixed history \(r\), equation \eqref{eq:L-diamond-bound} implies
\begin{equation}
 \diamondnorm{Q_{x,r,t}-Q_{y,r,t}}\le4\lambda(t).
 \label{eq:pairwise-generator}
\end{equation}
Even after adjoining a reference \(Z\), the accessible history register is classical. A reachable state can therefore be written as \(\sum_rP_r\otimes\rho_{r,SZ}\), with positive blocks whose traces sum to one. On block \(r\), equation \eqref{eq:pairwise-generator} bounds the trace norm of the generator difference applied to that block by \(4\lambda(t)\Tr\rho_{r,SZ}\). The trace norm adds over the diagonal blocks, giving the bound \(4\lambda(t)\) for the full state.

Apply variation of constants between interventions. The free generator is the same for inputs \(x\) and \(y\), so only the oracle terms remain in the generator difference. Propagation after each insertion, including later interventions, is CPTP and contracts trace distance. Integrating the bound \eqref{eq:pairwise-generator} and telescoping at the intervention times gives, for channels from the initialized history space,
\begin{equation}
 \diamondnorm{\mathcal A_x-\mathcal A_y}\le4\tau.
 \label{eq:pairwise-channel}
\end{equation}
Choose inputs with \(f(x)\ne f(y)\). Success probability \(2/3\) makes their output distributions differ in total variation by at least \(1/3\), so the output states differ in trace norm by at least \(2/3\). It follows that \(\tau\ge1/6\). Thus \(1+\tau\le7\tau\) in \eqref{eq:query-from-action}, and taking the infimum gives \(Q(f)=O(MQ(f))\).

Conversely, replace each oracle call in a \(Q(f)\)-query circuit by a Hamiltonian pulse \(H_x=O_x\) of duration \(\pi/2\). Since
\begin{equation}
 e^{-i(\pi/2)O_x}=-iO_x,
 \label{eq:oracle-pulse}
\end{equation}
each pulse has the effect of one oracle call apart from a global phase. Insert the free circuit gates as discrete interventions between pulses. The resulting protocol is piecewise constant, obeys the required regularity conditions, and has action \((\pi/2)Q(f)\). This proves \(MQ(f)\le(\pi/2)Q(f)\); the equality with the adversary bound follows from its characterization of quantum query complexity for finite functions~\cite{LMRSS11}.
\end{proof}

Continuous-to-discrete conversions and adversary bounds for Hamiltonian oracle models were established in~\cite{CGMSY09,YongeMallo11}, and constant-factor equivalence for general state conversion follows from~\cite{LMRSS11}. Here the permitted oracle couplings also include dissipators, and the protocol may choose oracle-independent CPTP controls using its measurement record. Corollary~\ref{cor:adversary} shows that the general adversary bound remains a lower bound, up to a constant factor, on the action required in this setting.

The output in our approximation is the accessible system together with its classical record. Specifying this channel leaves the emitted bath state and the trajectory unraveling undetermined. Bath monitoring can be included only after choosing an instrument whose implementation meets the cell conditions. Reinteraction with old baths, indefinite causal order~\cite{Abbott24}, unbounded generators, expected-action accounting, and uncharged postselection are outside this model.

\section{Local lattice dynamics and occupied-input compilation}\label{loc:sec:cells}

We now impose spatial locality and charge the elementary gates used by the compiler.
The block encodings in this part are built from coherent local-matrix evaluation rather
than supplied as unit-cost oracles. Each regional sweep visits a given bath only once,
but consecutive sweeps act on a common bath and may receive occupied inputs. The
extension in Lemma~\ref{loc:lem:sparse} treats those inputs. These shared baths are simulation
registers used to compose a dilation; their reuse does not enlarge the physical
adaptive-protocol model of Section~\ref{sec:adaptive}.

In this part, \(\Lambda_{\rm loc}\) is a per-term normalization bound, and \(t_\star\) denotes
a fixed short segment length. The total action in the general query model remains
\(\tau\) as defined in Sections~\ref{sec:time-dependent} and~\ref{sec:adaptive}.

Let the lattice be a finite $n$-site box in $\mathbb Z^D$, with $q$-dimensional sites,
fixed $D$, and open or periodic boundaries. Write $Z_e$ for the support of $e\in E$, and
require
\[
1\le |Z_e|\le s_0,\qquad \operatorname{diam}_\infty Z_e\le\rho.
\]
Here periodic intervals may cross a boundary. We fix $q,s_0,\rho,g,D$, take $\rho\ge1$,
and permit no more than $g$ supports at a site. Hence $|E|\le gn$. Supports of size one
and zero-valued terms are included.

The generator is
\begin{equation}\label{loc:eq:generator}
\begin{split}
\cL(t)&=\sum_{e\in E}\cL_e(t),\\
\cL_e(t)(\rho)&=-i[H_e(t),\rho]\\
&\quad+\sum_{\mu=1}^{r_{\rm j}}
\left(L_{e\mu}(t)\rho L_{e\mu}(t)^\dagger
-\tfrac12\{L_{e\mu}(t)^\dagger L_{e\mu}(t),\rho\}\right).
\end{split}
\end{equation}
The jump count $r_{\rm j}$ is fixed, and common positive rational bounds satisfy
\[
\norm{H_e(t)}\le\alpha_H,\qquad
\norm{L_{e\mu}(t)}\le\alpha_\mu.
\]
Put
\[
\alpha_L^2=\sum_\mu\alpha_\mu^2,\qquad
\Lambda_{\rm loc}=\alpha_H+\alpha_L^2.
\]
These normalization bounds do not vary with the support or sampled time. Let $\mathcal
E(t,s)$ denote the propagator of $\cL(t)$.

For each $e$, allow at most $B$ breakpoints in $[0,T]$, and assume that on every
interval between them,
\begin{equation}\label{loc:eq:holder}
\norm{\cL_e(t)-\cL_e(s)}_\diamond
\le K_t|t-s|^\alpha,\qquad 0<\alpha\le1.
\end{equation}
Fix $\alpha$. We need neither a common set of breakpoints nor their locations as input.
Values at a breakpoint can be arbitrary subject to the normalization bounds.

We impose regularity on the generator itself. Operator-norm H\"older bounds for the
Hamiltonian and jump matrices suffice, by
\begin{align*}
\norm{-i[H-H',\,\cdot\,]}_\diamond&\le2\norm{H-H'},\\
\norm{\mathcal D[L]-\mathcal D[L']}_\diamond
&\le2(\norm L+\norm{L'})\norm{L-L'}.
\end{align*}
The argument requires no higher time derivatives.

The reversible evaluator for a term takes a binary sample address and returns matrix
entries to error $2^{-b}$. The uniform cost bound is polynomial in $b$ and the temporal
and spatial address lengths; time arithmetic and coherent lookup, if used, are charged.
Selecting $p$ patch descriptions by multiplexing costs at most $p$ evaluations.

Every operator norm here is spectral. Identity factors outside a bath projector's
specified registers are implicit. Reference systems are allowed in the operator-norm
estimates; our constants may depend on fixed local parameters, $\alpha$, and uniform
evaluation algorithms.

\begin{theorem}[Time-dependent lattice simulation]\label{loc:thm:main}
For the class above, let $T\ge0$ and $0<\eps<1/8$. There is a circuit implementing a
channel $\widetilde{\mathcal E}$ with
\[
\norm{\widetilde{\mathcal E}-\mathcal E(T,0)}_\diamond\le\eps
\]
using \eqref{loc:eq:intro-bound} gates, depth $(T+1)\polylog X$, and $n\polylog X$ working
qubits if baths are discarded between short time segments. Geometrically neighboring
one- and two-qubit gates in the patch layout suffice.
\end{theorem}

Section~\ref{loc:sec:assembly} gives the gate count and an explicit overhead exponent.

\subsection{A fixed local dilation}\label{loc:sec:local-dilation}

For one local term at a sampled time, suppress the spatial and temporal labels.
Use the Euler column and its polar isometry from
\eqref{eq:Ah}--\eqref{eq:polar}, writing \(A_h=\mathcal A_h\) in this part and
\(\Gamma_L=L^\dagger L=\sum_\mu L_\mu^\dagger L_\mu\).
The bound \eqref{eq:polar-channel-error} is \(O(h^2)\) under the fixed local
normalizations. What must also be specified here is a full unitary extension whose
free gates admit efficient routing on occupied bath inputs.

For a constant-size contraction $A$, its local matrix also specifies the unitary
\[
\begin{pmatrix}
A&\sqrt{I-AA^\dagger}\\
\sqrt{I-A^\dagger A}&-A^\dagger
\end{pmatrix}.
\]
The signal block of this unitary is $A$, so choosing $H/\alpha_H$ or $L_\mu/\alpha_\mu$
gives the desired encoding. The dimension is constant. We may therefore evaluate its
positive square roots in polynomial bit cost, rescaling the rounded input slightly
inward to keep it a contraction. At error $\nu$, this gives $\poly(\log(1/\nu))$ gates
for the encoding, a control, or an adjoint.

Use the three-call correction of Lemma~\ref{lem:cubic-polar}. Its coefficient
states, with the common local normalizations, are
\begin{equation}\label{loc:eq:cell-coefficients}
\begin{split}
s_B&=1+h\alpha_H+\tfrac12h\alpha_L^2,\\
c&=s_B^2+h\alpha_L^2,\qquad s=c^{-1/2},\\
\ket{g_B}&=\frac{\ket I+\sqrt{h\alpha_H}\ket H
+\sqrt{h\alpha_L^2/2}\ket D}{\sqrt{s_B}},\\
\ket{g_A}&=\frac{s_B\ket0+\sqrt h\,\alpha_L\ket1}{\sqrt c}.
\end{split}
\end{equation}
Prepare the stack state $g_A$. On stack zero, use the role state $g_B$ to select
$I,-iU_H,-U_{\Gamma_L}$, and unprepare the role. On stack one, prepare the jump label and
select the individual jump encodings. All inactive branches act identically on their
work registers.

To encode the product $U_{\Gamma_L}$, we use a private jump label and two independent block
signals, keeping a separate label and signal for the retained jump branch. The output
projector fixes the role and all work signals, including the private product label,
while leaving the stack and retained jump label free. The input projector fixes those
retained labels as well. This gives
\begin{equation}\label{loc:eq:column-encoding}
\Pi_{\rm out}W_h\Pi_{\rm in}=sA_h,
\end{equation}
with zero rows added where needed. Keeping these registers separate preserves the
intermediate projection in the encoding of $L^\dagger L$ and specifies the signal
projectors used in Section~\ref{loc:sec:frames}.

For \(s\ge1/2\), take \(g_s,u_s\) from \eqref{eq:explicit-phases} and
\(\Phi_P(u)=I+(u-1)P\). The full cell is
\begin{equation}\label{loc:eq:full-cell}
u(h)=g_sW_h\Phi_{\Pi_{\rm in}}(u_s)
W_h^\dagger\Phi_{\Pi_{\rm out}}(u_s)W_h.
\end{equation}
Both scalar phases have modulus one.

\begin{lemma}[Local cell bounds]\label{loc:lem:cell}
There are constants $h_0,a,b,\kappa,c_0>0$, independent of the local term and sampled
time, for which the circuit \eqref{loc:eq:full-cell} satisfies
\begin{align}
u(0)P^0&=P^0,\label{loc:eq:cell-zero}\\
\norm{u(h)-u(0)}&\le a\sqrt h,\label{loc:eq:cell-weak}\\
\norm{P^0(u(h)-u(0))P^0}&\le bh,\label{loc:eq:cell-vac}\\
\norm{P^1u(h)P^0},\ \norm{P^0u(h)P^1}
&\le\kappa\sqrt h,\label{loc:eq:cell-off}\\
\norm{\mathcal J_h-e^{h\cL_e}}_\diamond&\le c_0h^2,\label{loc:eq:cell-channel}
\end{align}
for $0\le h\le h_0$. Here $P^1=I-P^0$, and $\mathcal J_h$ is the reduced channel from a
vacuum bath. Both $u(h)$ and its adjoint, with their respective slot orders, are weak
causal cells of weight $O(h)$.
\end{lemma}

\begin{proof}
The full-column estimate \eqref{eq:local-error} gives
\[
\norm{u(h)\Pi_{\rm in}-V_h}\le h^2\Lambda_{\rm loc}^2,
\qquad
\norm{\mathcal J_h-e^{h\cL_e}}_\diamond\le7h^2\Lambda_{\rm loc}^2.
\]
Lemma~\ref{lem:cubic-polar} controls the entire output column, including the
component outside the signal space.

The vacuum row of $V_h$ is $I+hK+O(h^2)$, which gives \eqref{loc:eq:cell-vac}, and its
nonvacuum rows have norm $O(\sqrt h)$. The two off-diagonal blocks of a unitary have
equal norm. Hence \eqref{loc:eq:cell-off} holds for either orientation.

Choose each preparation extension to be the minimal rotation from its reference. Its
deviation from identity is $O(\sqrt h)$; for the scalar phases, the deviation from the
zero-coupling value is $O(h)$. There are constantly many factors, so telescoping proves
\eqref{loc:eq:cell-weak}. When the coupling vanishes, the initialized path takes the
identity role alone. On that role the phases of \eqref{loc:eq:full-cell} have product one,
proving \eqref{loc:eq:cell-zero}.

For the weak causal estimates, replace each free gate by its zero-preparation version
without changing the slots, projectors, or scalar phases. Along this comparison path the
state is clean and never enters an active port. There are only constantly many gates to
telescope, giving $O(\sqrt h)$ bounds on active entry, absorbed clean return, and
nonclean output along the all-inactive path. The reversed calculation applies to
$u(h)^\dagger$.
\end{proof}

Only the vacuum is fixed by $u(0)$. On $P^1$ there may be scattering, different at
different sample times. We keep this part of the cell in the locality and occupation
estimates instead of assuming that zero coupling makes the full unitary the identity.

\subsection{Compilation on occupied inputs}\label{loc:sec:sparse}

A finite sweep uses distinct cell baths $E_1,\ldots,E_M$. Cell $j$ acts on the system
and $E_j$, and later cells leave $E_j$ unchanged. Define
\[
\cN=\sum_{j=1}^M P_j^1,\qquad
P_{\le N}=\mathbf1_{[0,N]}(\cN),\qquad
\Omega=\sum_jw_j.
\]
The occupied locations may be in superposition and entangled with the system.

We use Definition~\ref{def:cell}: each cell has a fixed number of
oracle slots, and its active-entry maps, direct sum of actual query-input maps, absorbed
clean-return maps, and all-inactive nonclean column have norms $O(\sqrt{w_j})$. In the
deformed circuit, every actual oracle call receives a factor $z$, including calls to an
adjoint.

\begin{lemma}[Sparse-input compiler]\label{loc:lem:sparse}
A sweep of weak causal cells can be approximated on $\ran P_{\le N}$ to operator error
$\delta$ using
\[
O\bigl(1+N+\Omega+\log(1/\delta)\bigr)
\]
controlled master queries. The additional compiler work is reset in the target isometry.
Its free operations consist of transducer calls, coefficient preparations, and
reflections about the specified input subspace.
\end{lemma}

\begin{proof}
Let $T_j(z)$ denote the deformed full cell. On $|z|=2$, its clean column has Gram norm
at most $e^{Cw_j}$, both off-diagonal blocks are $O(\sqrt{w_j})$, and the full norm is
bounded by a constant. For the reverse off-diagonal block, expand over nonempty slot
subsets and apply the absorbed clean-return condition. Subtracting these contributions
from $T_j(1)$ bounds the empty subset.

Choose a fixed $d>1$ large enough that
\[
D_j=P_j^0+dP_j^1,\qquad
\overline T_j(z)=T_j(z)D_j^{-1}
\]
has occupied-column Gram block at most $1/4$. The cross Gram block is $O(\sqrt{w_j})$,
so it can be absorbed into the margin in the occupied sector. We obtain
\begin{equation}\label{loc:eq:sparse-weight}
\norm{\overline T_j(z)}\le e^{Cw_j}.
\end{equation}
If $w_j$ exceeds a fixed constant, the unrestricted cell bound gives the same estimate
after increasing $C$.

Let $K_j(z)$ be the direct sum of query-input maps within cell $j$. Its clean
restriction is $O(\sqrt{w_j})$, and its unrestricted norm is $O(1)$. Put $D_{\rm
all}=\prod_jD_j$. For $x\in\ran P_{\le N}$, write $y=D_{\rm all}x$. Then
\[
\norm y\le d^N\norm x,\qquad
\sum_j\norm{P_j^1y}^2\le N\norm y^2.
\]
The history component ending in cell $j$ is
\[
K_jD_j^{-1}
\left(\prod_{i<j}^{\leftarrow}\overline T_i\right)
\left(\prod_{i>j}D_i^{-1}\right)y.
\]
The preceding cells commute with $P_j^1$, and the remaining weights are contractions.
Summing squared norms over the orthogonal query ports therefore gives
\begin{equation}\label{loc:eq:sparse-history}
\norm{\Gamma(z)P_{\le N}}
\le C d^N e^{C\Omega}(\sqrt\Omega+\sqrt N),
\qquad |z|=2.
\end{equation}

Write the one-query transducer as
\[
S=\begin{pmatrix}D&C\\B&A\end{pmatrix}.
\]
Its graph relation and chronological history expansion are
\begin{equation}\label{loc:eq:graph}
S\begin{pmatrix}P\\\Gamma P\end{pmatrix}
=\begin{pmatrix}VP\\\Gamma P\end{pmatrix},
\qquad
\Gamma(z)P=\sum_{r\ge0}z^rA^rBP,
\end{equation}
where $P=P_{\le N}$, including the other initialized work. The sum is finite. Cauchy's
estimate at radius two gives
\[
\norm{A^k\Gamma(1)P}
\le C2^{-k}d^Ne^{C\Omega}(\sqrt\Omega+\sqrt N).
\]
Taking $k=O(1+N+\Omega+\log(1/\delta))$ makes this at most $\delta/4$.

Theorem~\ref{thm:tail-compiler} now gives error at most \(\delta\) with at most
\(9k\) transducer calls. This is the same two-length reuse and rectangular
amplification construction proved in Section~\ref{sec:compiler}, with the source
projector enlarged from the vacuum condition to \(P\). The source reflection includes
\(P\), while the output reflection fixes only the added compiler work.
\end{proof}

Reversing the sweep and taking the adjoint of each cell gives the same bound for the
inverse on its own sparse input subspace. We use this construction for each backward
regional factor.

\section{Spatial decomposition on a common bath}\label{loc:sec:locality}

Divide a short interval $[a,a+t]$ into $J$ equal bins of width $h=t/J$, freezing the
local matrices in bin $j$ at
\[
\theta_j=a+(j-\tfrac12)h,\qquad 1\le j\le J.
\]
Give $(e,j)$ its own fresh bath. In each bin use one fixed order of interactions; the
restriction to $R\subseteq E$ is denoted $U_{R,J}$. A regional restriction changes only
which terms are applied, leaving the sampled matrices, unitary extensions, and bath
registers in common.

\subsection{Temporal discretization}\label{loc:sec:temporal-mesh}

First consider a common bin width $h$ for the full evolution, and write $\mathcal
J_h(T)$ for the reduced channel of the forward sweep. We will form the short time
segments by grouping these bins.

\begin{lemma}[Freezing and local-cell error]\label{loc:lem:temporal-mesh}
For $h\le h_0$, the finite channel satisfies
\begin{equation}\label{loc:eq:mesh-channel}
\norm{\mathcal J_h(T)-\mathcal E(T,0)}_\diamond
\le C\bigl(n^2Th+nK_tTh^\alpha+nBh\bigr).
\end{equation}
Neither mesh size nor breakpoint positions enter the constant.
\end{lemma}

\begin{proof}
Write $\cL_h(s)$ for the midpoint-frozen generator. The variation-of-constants identity reads
\[
\begin{split}
\mathcal E(T,0)-\mathcal E_h(T,0)
={}&\int_0^T\mathcal E(T,s)
[\cL(s)-\cL_h(s)]\\
&\hspace{3.4em}\cdot\mathcal E_h(s,0)\,ds.
\end{split}
\]
The two propagators here are CPTP, hence each has diamond norm one.

Fix a local term. Bins whose closures meet one of its breakpoints cover at most $2Bh$
time; call them bad. On a good bin use $K_th^\alpha$ from \eqref{loc:eq:holder}, and on a
bad bin use the constant normalization bound. The time integral, summed over $|E|=O(n)$
terms, is $O(nK_tTh^\alpha+nBh)$. Breakpoints are counted over the entire evolution for
each term, rather than recounted after dividing time into segments.

For the frozen evolution, the Lie-product error is $O(n^2h^2)$ per bin, and replacing
its local channels by the cells of Lemma~\ref{loc:lem:cell} adds $O(nh^2)$. All these maps
are CPTP. Telescoping over $T/h$ bins therefore gives $O(n^2Th)$, completing the
estimate.
\end{proof}

\subsection{Folding nonstationary passes}\label{loc:sec:transfer}

Let $v_{s,j}$ act on a $d$-dimensional system and bin bath $j$, and write
\[
U_s=v_{s,J}\cdots v_{s,1},\qquad
W_J=U_r^{\eta_r}\cdots U_1^{\eta_1},
\quad \eta_s\in\{+1,-1\}.
\]
Give each pass its own system copy. On bin $j$, put
\[
\widehat v_{s,j}=
\begin{cases}
v_{s,j},&\eta_s=+1,\\
(v_{s,j}^\dagger)^{T_S},&\eta_s=-1,
\end{cases}
\]
where the transpose acts only on system indices. Define
\[
M_j=\bra0\widehat v_{r,j}^{(S_rE_j)}
\cdots\widehat v_{1,j}^{(S_1E_j)}\ket0.
\]
For simple tensors, set
\[
\cC_\eta(A_1\otimes\cdots\otimes A_r)
=A_r^{\sharp_r}\cdots A_1^{\sharp_1},
\]
where $\sharp_s$ is the identity for a forward pass and transpose for a backward pass.
Extend this map linearly.

\begin{lemma}[Finite-pass transfer]\label{loc:lem:transfer}
For every finite $J$, including different cells in different bins,
\begin{equation}\label{loc:eq:transfer}
\bra{\vac}W_J\ket{\vac}
=\cC_\eta(M_J\cdots M_1).
\end{equation}
\end{lemma}

\begin{proof}
Use system matrix units to expand a local unitary. Transposition for a backward pass
reverses their multiplication while preserving the order of factors on distinct baths.
With one system copy for each pass, different pass/bin factors can be commuted into bin
order. Vacuum contraction then gives $M_J\cdots M_1$, and joining the remaining system
indices in their physical order is precisely $\cC_\eta$. The cells can be chosen
independently in every bin.
\end{proof}

Inside a bin, distinct interactions again have distinct baths, so the oriented sweep
factors into cells by the same argument. These system copies are contraction indices
only; the simulation circuit does not allocate replicated systems.

For $r=2v$ alternating signs, write $\cC_r$ for the contraction. Set
\[
\ket\Omega=\sum_{i=1}^d\ket{ii},\qquad
\ket\Phi=\ket\Omega_{12}\cdots\ket\Omega_{2v-1,2v}.
\]
For $w=\bra\Phi R$,
\[
[\cC_r(R)]_{ki}
=\sum_{j_1,\ldots,j_{v-1}}
w_{i\,j_1\,j_1\,\cdots\,j_{v-1}\,j_{v-1}\,k}.
\]
Cauchy--Schwarz on these $d^{v-1}$ terms gives
\begin{equation}\label{loc:eq:boundary}
\begin{split}
\norm{\cC_r(R)}
&\le\norm{\cC_r(R)}_{\rm F}\\
&\le d^{(v-1)/2}\norm w
\le d^{(r-1)/2}\norm R.
\end{split}
\end{equation}
Here $\norm\Phi=d^{v/2}$, and $\cC_r(I)=I$. On a local support, only its system
dimension is used in \eqref{loc:eq:boundary}.

The exponent cannot be improved for a general contraction of this form. Take
$R=\ket{\Phi/\norm\Phi}\bra z$, with $z$ the normalized equal superposition of the index
strings appearing in the sum for a fixed pair $i,k$. Then $\norm R=1$, while that output
entry is $d^{(r-1)/2}$.

\subsection{One factor of \texorpdfstring{$h$}{h} per interaction}\label{loc:sec:support}

Consider $r$ regional passes, where $r$ is fixed. Expand each oriented local factor as
its zero-coupling value plus its difference. Group terms by the set $S\subseteq E$ on
which at least one difference was selected:
\[
M_j=I+\sum_{\varnothing\ne S\subseteq E}M_{j,S}.
\]
Write $V(S)=\bigcup_{e\in S}Z_e$.

\begin{lemma}[Finite-bin support bound]\label{loc:lem:support}
There is a constant $c_r$, uniform in the sampled times, such that
\begin{equation}\label{loc:eq:one-bin-bound}
\supp M_{j,S}\subseteq V(S),\qquad
\norm{M_{j,S}}\le(c_rh_j)^{|S|}.
\end{equation}
The statement allows unequal bin widths $h_j$. If $R_{J,S}$ collects all terms in
$M_J\cdots M_1$ with union support exactly $S$, then
\begin{equation}\label{loc:eq:support-powers}
\norm{R_{J,S}}
\le\left(\prod_j(1+c_rh_j)-1\right)^{|S|}
\le(e^{c_rt}-1)^{|S|},
\end{equation}
where $t=\sum_jh_j$.
\end{lemma}

\begin{proof}
If no difference is selected on an interaction, its bath remains in the vacuum and every
zero-coupling factor acts as the identity there. We can remove that interaction from the
product, even when its zero-coupling action on the nonvacuum space varies with time.

If exactly one difference is selected, commute the corresponding bath-vacuum projectors
through the other interactions. The zero-coupling factors preserve both vacuum and its
complement, leaving only the vacuum-to-vacuum block of the selected difference. This has
norm at most $bh_j$ by Lemma~\ref{loc:lem:cell}. With two or more differences, at least two
weak factors are already present.

Let $d_0=q^{s_0}$, and use the constants $a,b,h_0$ of Lemma~\ref{loc:lem:cell}. Partial
transpose on one local system increases operator norm by at most $d_0$. It is sufficient
to take
\[
c_r=rd_0b+
d_0^r\sum_{u=2}^r\binom ru a^u h_0^{u/2-1},
\]
enlarged to at least one. This accounts for every choice of selected differences and the
intervening zero-order factors. The bath projectors belong to distinct tensor factors,
so the estimates apply simultaneously on all selected terms. Their system support is
$V(S)$.

For the ordered transfer product, we bound each product of operators by the product of
its norms. The resulting count of exact union supports is the scalar product
\[
\prod_j\prod_{e\in E}(1+c_rh_jz_e),\qquad z_e^2=z_e.
\]
Selecting a given interaction at least once has coefficient $\prod_j(1+c_rh_j)-1$. These
choices are independent for distinct interaction labels, giving
\eqref{loc:eq:support-powers}.
\end{proof}

\subsection{A complete spatial word}\label{loc:sec:layers}

Choose an integer spacing $\ell>\rho$. Along each open coordinate, place cuts between
sites at spacing $\ell$, and color them alternately left and right. For coordinate $a$,
let $E_A^{(a)}$ exclude the terms crossing right cuts, let $E_C^{(a)}$ exclude those
crossing left cuts, and put $E_B^{(a)}=E_A^{(a)}\cap E_C^{(a)}$.

For periodic coordinates of length $L_a\ge2\ell$, use
$2\lfloor L_a/(2\ell)\rfloor$ cuts with integer spacings between $\ell$ and $2\ell$, alternating cyclically. If $L_a<2\ell$, use no cut in that coordinate. A local support crosses a cut when its specified coordinate interval or arc crosses that cut.

For a retained interaction set $R$, define
\begin{equation}\label{loc:eq:recursive-word}
\begin{split}
\mathcal W_0(R)&=U_{R,J},\\
\mathcal W_a(R)&=
\mathcal W_{a-1}(R\cap E_A^{(a)})\\
&\quad\cdot\mathcal W_{a-1}(R\cap E_B^{(a)})^\dagger
\mathcal W_{a-1}(R\cap E_C^{(a)}).
\end{split}
\end{equation}
Every factor in this definition acts on the original common bath.

\begin{lemma}[Spatial layout]\label{loc:lem:layout}
The expanded word $\mathcal W_D(E)$ has $r_0=3^D$ passes with alternating forward and
backward orientations, beginning and ending forward. Each pass is a product of unitaries
on disjoint boxes of side at most $4\ell$. Every individual bath is visited at most
$r_0$ times.
\end{lemma}

\begin{proof}
Taking an adjoint reverses the order of a subword and changes every orientation. Thus an
odd alternating word that begins and ends forward becomes one that begins and ends
backward. Induction in \eqref{loc:eq:recursive-word} gives the claimed $3^D$-pass pattern.

Each pass retains the intersection of one set from $E_A^{(a)},E_B^{(a)},E_C^{(a)}$ in
every coordinate. The remaining cuts bound intervals of length at most $2\ell$ for open
boundaries and $4\ell$ for periodic boundaries, so each retained support lies in one of
the resulting boxes. Different boxes have disjoint system and bath registers. Their
sweeps therefore commute and can be run in parallel, with each bath visited at most once
in the pass.
\end{proof}

We will estimate the overlap $U_{E,J}^\dagger\mathcal W_D(E)$ directly. It has $r=3^D+1$
alternating passes, so the preceding transfer bounds apply to the whole product.

\subsection{Connected-support cancellation}\label{loc:sec:cuts}

Connect two labels in $E$ when their supports intersect. This interaction graph has degree at most
\[
\Delta_I\le s_0(g-1).
\]
For $Q\subseteq E$, define the restricted physical overlap
\[
F(Q)=\bra{\vac}U_{Q,J}^\dagger\mathcal W_D(Q)\ket{\vac}.
\]
Replacing all other cells by zero-coupling factors gives the same vacuum overlap as
omitting them. Hence
\begin{equation}\label{loc:eq:mobius}
F_S=\cC_r(R_{J,S})
=\sum_{Q\subseteq S}(-1)^{|S|-|Q|}F(Q).
\end{equation}

\begin{lemma}[Cancellation on small supports]\label{loc:lem:cancel}
Exact-support coefficients factor over the connected components of the interaction graph
induced by $S$. A nonempty connected coefficient vanishes if, in every coordinate, $S$
misses at least one cut color. In particular,
\[
\begin{gathered}
F_S=0\quad\text{for connected }S\text{ with }|S|<m_\ell,\\
m_\ell=\lfloor\ell/\rho\rfloor+1.
\end{gathered}
\]
\end{lemma}

\begin{proof}
Different connected components have disjoint physical supports and baths. The restricted
words and their vacuum overlaps therefore factor over these components, as does the
subset sum in \eqref{loc:eq:mobius}.

Suppose $S$ misses the right cuts in coordinate $a$. Then
\[
S\cap E_A^{(a)}=S,\qquad
S\cap E_B^{(a)}=S\cap E_C^{(a)}.
\]
The last two subwords in \eqref{loc:eq:recursive-word} are mutual inverses and cancel,
leaving $\mathcal W_{a-1}(S)$. If $S$ misses the left cuts, the first two cancel.
Repeating this in every coordinate gives $\mathcal W_D(S)=U_{S,J}$. Every subset of $S$
has the same property, so its nonempty M\"obius difference vanishes.

For $k$ connected local terms, the union of their coordinate intervals is connected and
has length at most $k\rho$. This also holds for local arcs in a periodic coordinate
until their union covers the circle. Opposite cut colors are separated by at least
$\ell$, and crossing both requires length strictly greater than $\ell$, since the cuts
lie between sites. Covering the periodic circle requires at least this length as well.
Hence a contributing support must satisfy $k\rho>\ell$.
\end{proof}

\begin{proposition}[Uniform lattice decomposition]\label{loc:prop:layers}
Let $r=3^D+1$, and put
\[
\begin{gathered}
\chi=q^{(r-1)/2},\qquad A_I=\max\{1,\Delta_I^2\},\\
\xi=A_I\chi^{s_0-1}(e^{c_rt}-1).
\end{gathered}
\]
For $\xi<1$, set
\[
B_*=\frac{|E|\chi\,\xi^{m_\ell}}{1-\xi}.
\]
Then, uniformly in the sampled operator values and the finite mesh,
\begin{equation}\label{loc:eq:layer-bound}
\norm{(U_{E,J}-\mathcal W_D(E))P_{\vac}}
\le\sqrt{2(e^{B_*}-1)}.
\end{equation}
\end{proposition}

\begin{proof}
A connected set of $k$ interactions covers at most $1+(s_0-1)k$ sites: start with one
support and add the others along a spanning tree, with each new support intersecting its
predecessor. Equations \eqref{loc:eq:boundary} and \eqref{loc:eq:support-powers} therefore give
\[
\norm{F_S}
\le\chi\left[\chi^{s_0-1}(e^{c_rt}-1)\right]^k.
\]

There are at most $|E|A_I^k$ connected sets of $k$ interactions. Choose the least label
in a set as its root and fix a canonical spanning tree. A depth-first traversal has
length $2(k-1)$ and visits exactly that set, with at most $\max\{1,\Delta_I\}$ choices
at each step. Counting these walks bounds the number of sets.

Lemma~\ref{loc:lem:cancel} removes all components except those with $k\ge m_\ell$, whose
norms sum to at most $B_*$. The coefficients for disconnected supports factor into their
connected components. Dropping the disjointness restriction from the sum over
collections of components gives
\[
\norm{I-F(E)}
\le\prod_{\text{connected }S}(1+\norm{F_S})-1
\le e^{B_*}-1.
\]
Since both physical words are unitary on the common bath,
\begin{equation}\label{loc:eq:isometry-overlap}
\begin{split}
\norm{(U_{E,J}-\mathcal W_D(E))P_{\vac}}^2
&=\norm{2I-F(E)-F(E)^\dagger}\\
&\le2\norm{I-F(E)}.
\end{split}
\end{equation}
which proves \eqref{loc:eq:layer-bound}.
\end{proof}

Choose a fixed ${t_\star}>0$ such that
\[
A_I\chi^{s_0-1}(e^{c_r{t_\star}}-1)\le\tfrac12.
\]
For times at most ${t_\star}$, joint error $\eta$ is then obtained with
\[
\ell=O\bigl(\rho[1+\log(n/\eta)]\bigr).
\]
Each region contains $O(\ell^D)$ sites and local terms. The estimate applies to the
complete $3^D$-pass word on the initial vacuum, so the occupied baths encountered inside
that word require no separate vacuum-input approximation.

On a nearest-neighbor chain, connected sets are intervals, with at most $n$ choices at
each length. Thus $D=1$ gives the three-pass construction and allows $A_I=1$ in the
count. In higher dimensions, the bounded-degree interaction graph supplies the
corresponding estimate.

\section{Occupation at regional boundaries}\label{loc:sec:physical-occupation}

The bath starts in the vacuum, but earlier regional factors may leave it occupied. To
use the sparse-input compiler at each subsequent factor, we need a bound on the
occupation at its input.

Write one cell relative to its bath vacuum as
\[
u=\begin{pmatrix}A&B\\C&D\end{pmatrix},
\qquad
\norm B,\norm C\le\kappa\sqrt h.
\]
For $d\ge1$, put $D_d=P^0+dP^1$. Then
\[
D_duD_{2d}^{-1}
=\begin{pmatrix}A&B/(2d)\\dC&D/2\end{pmatrix}.
\]

\begin{lemma}[Weighted physical cell]\label{loc:lem:physical-weight}
The unitary cell satisfies
\begin{equation}\label{loc:eq:physical-weight}
\norm{D_duD_{2d}^{-1}}
\le\exp\left(\tfrac23d^2\kappa^2h\right).
\end{equation}
\end{lemma}

\begin{proof}
Unitarity gives $A^\dagger B=-C^\dagger D$. The upper Gram block is at most
$I+d^2\kappa^2hI$, the lower one is at most $I/4$, and the cross block has norm at most
$d\kappa\sqrt h/2$. For its quadratic form, absorb the cross term with
\[
2\beta xy\le\tfrac34y^2+\tfrac43\beta^2x^2.
\]
The Gram norm is at most $1+\tfrac43d^2\kappa^2h$. Taking the square root and using
$1+x\le e^x$ proves the claim.
\end{proof}

Fix an interaction label $e$, and let $\cN_e$ count its occupied time-bin baths. The
spatial word visits each bin at most $r_0=3^D$ times. Assign an initial weight that can
be halved at each visit and still finish at two; the largest output weights needed are
$2^{r_0},\ldots,4,2$. All operations on other baths commute with these weights,
including those whose system supports overlap.

Telescoping \eqref{loc:eq:physical-weight} through any physical prefix gives
\begin{equation}\label{loc:eq:physical-tail}
\begin{split}
\norm{2^{\cN_e}V_{\rm pre}P_{\vac}}&\le e^{C_{\rm occ}},\\
\norm{\mathbf1_{\cN_e>k}V_{\rm pre}P_{\vac}}
&\le2^{-k}e^{C_{\rm occ}},
\end{split}
\end{equation}
with
\[
C_{\rm occ}=\tfrac23\kappa^2{t_\star}\sum_{s=1}^{r_0}4^s
=\tfrac89\kappa^2{t_\star}(4^{r_0}-1).
\]
uniformly in the mesh. The initial weights act as the identity on vacuum, and fewer
visits can only lower the estimate. For $D=1$, this constant is $56\kappa^2{t_\star}$.

For a region with at most $p$ interaction labels, define
\[
P_{\rm good}=\prod_{e\text{ in region}}\mathbf1_{\cN_e\le k}.
\]
The projectors commute, and
\[
I-P_{\rm good}\preceq
\sum_{e\text{ in region}}\mathbf1_{\cN_e>k}.
\]
Consequently
\begin{equation}\label{loc:eq:regional-tail}
\norm{(I-P_{\rm good})V_{\rm pre}P_{\vac}}
\le\sqrt p\,2^{-k}e^{C_{\rm occ}}.
\end{equation}
On the good subspace, the regional bath has occupation at most $N_0=pk$.
Lemma~\ref{loc:lem:sparse} therefore applies at every regional boundary. This estimate is
coherent and uniform over the input system and its reference.

\section{A gate implementation of the transducer}\label{loc:sec:router}

A cell's free gates touch only a constant-size ancilla; its controlled SELECT is the
sole operation on the system. We exploit this form to avoid a sequential traversal of
all $M$ cells. Common normalization bounds make the free gates identical at every
sample. The addressed SELECT contains all the time dependence.

\subsection{Explicit singular frames}\label{loc:sec:frames}

Let $R_A,R_B$ be the preparation rotations for $g_A,g_B$. If every active SELECT in
$W_h$ is replaced by its inactive projector, only stack zero and the identity role
remain. Writing
\[
\ket l=\ket0\otimes R_B^\dagger\ket I,\qquad
\ket r=R_A^\dagger\ket0\otimes R_B^\dagger\ket I,
\]
gives the all-inactive block
\begin{equation}\label{loc:eq:dw}
D_W=\ket l\bra r\otimes I_{\rm aux}.
\end{equation}
where the auxiliary space contains the retained jump label and the separate work
registers of Section~\ref{loc:sec:local-dilation}.

Let $P_{\rm in}^{\rm a}$ fix every auxiliary register and let $P_{\rm out}^{\rm a}$ fix
every work register but leave the retained jump label free. They are commuting
projectors with $P_{\rm in}^{\rm a}\preceq P_{\rm out}^{\rm a}$. Since
\[
|\langle0,I|r\rangle|^2=s_B/c,\qquad
|\langle0,I|l\rangle|^2=1/s_B,
\]
the all-inactive block of the full corrected cell is
\begin{equation}\label{loc:eq:dcell}
D_{\rm cell}=\ket l Z\bra r,
\end{equation}
where
\[
\begin{split}
Z={}&g_s\left(I+(u_s-1)\frac{s_B}{c}P_{\rm in}^{\rm a}\right)\\
&\hspace{1.1em}\cdot
\left(I+(u_s-1)\frac1{s_B}P_{\rm out}^{\rm a}\right).
\end{split}
\]
This follows by substituting \eqref{loc:eq:dw} into the all-inactive product
\[
g_sD_W\Phi_{\rm in}D_W^\dagger\Phi_{\rm out}D_W.
\]

Since $Z$ is scalar on three known auxiliary sectors, we can choose its singular frames
explicitly. Take the minimal selector rotations from the reference state to $r$ and $l$,
placing the phases of the three scalars in the latter rotation. Then
\begin{equation}\label{loc:eq:local-svd}
D_{\rm cell}=U_L\,\diag(c_\beta)\,U_R^\dagger,
\qquad 0\le c_\beta\le1.
\end{equation}
with the nonzero $c_\beta$'s equal to the moduli of these scalars. The remaining
selector directions have singular value zero.

At $h=0$, the scalar on the full vacuum is $g_1u_1^2=1$, so these frames satisfy
\begin{equation}\label{loc:eq:public-frame-weak}
\begin{gathered}
U_\star(0)\ket0=\ket0,\qquad \star=L,R,\\
\norm{U_\star(h)-U_\star(0)}=O(\sqrt h).
\end{gathered}
\end{equation}
For sufficiently small $h_0$, the nonzero sector scalars remain bounded away from zero
and their phases can be chosen continuously. This explicit choice gives
\eqref{loc:eq:public-frame-weak} even at degenerate singular values, where an arbitrary
singular-vector routine would not ensure it.

Let $S_j^\circ$ be the local free transducer, with a private port for each of the three
macro SELECT calls. Each macro uses a constant number of supplied local encodings. After
padding unused private directions by identities, the private dimension is at least the
public ancilla dimension, and a cosine--sine decomposition consistent with
\eqref{loc:eq:local-svd} gives
\begin{equation}\label{loc:eq:local-csd}
S_j^\circ=
(U_{L,j}\oplus V_{L,j})\,
\mathsf C_j\,
(U_{R,j}^\dagger\oplus V_{R,j}^\dagger).
\end{equation}
For a bath basis state $\beta$, the factor $\mathsf C_j$ rotates the public rail with
one private bright rail, with cosine $c_\beta$. It preserves $\beta$ and fixes the other
private directions.

We can complete this decomposition directly from the local blocks. If $B,C$ are the
lower-left and upper-right blocks of $S_j^\circ$, the bright columns for $c_\beta<1$ are
\[
\frac{BU_R\ket\beta}{\sqrt{1-c_\beta^2}},
\qquad
-\frac{C^\dagger U_L\ket\beta}{\sqrt{1-c_\beta^2}}.
\]
Unitarity makes these columns orthonormal. Choose the remaining left private columns in
their orthogonal complement and apply the adjoint of the lower-right block to obtain the
right columns. This also covers $c_\beta=1$ and the dark subspace.

All matrices in this local free transducer have fixed dimension. Their entries are fixed
algebraic functions of $\sqrt h$, determined by the common normalizations rather than
the sampled Hamiltonian and jump matrices. On a sufficiently small positive interval,
each nonzero pivot in the construction has finite power-law order in $h$; zero sectors
can be treated exactly. Evaluating the matrices to error $\nu$ and synthesizing the
fixed-size gates therefore costs $\poly(\log(1/h)+\log(1/\nu))$. Throughout this
completion, we keep the public frames chosen in \eqref{loc:eq:public-frame-weak}.

\subsection{Changing frames across a sweep}\label{loc:sec:cascade}

Suppose the regional sweep has distinct baths $1,\ldots,M$. Define public frames
\[
E_{\rm in}=\bigotimes_jU_{R,j},\qquad
E_{\rm out}=\bigotimes_jU_{L,j}.
\]
On private port $j$, use
\begin{equation}\label{loc:eq:private-frame}
E_{{\rm p},j}=
\left(\bigotimes_{i<j}U_{L,i}\right)
\otimes V_{L,j}\otimes
\left(\bigotimes_{i>j}U_{R,i}\right),
\end{equation}
and put $E_{\rm p}=\bigoplus_jE_{{\rm p},j}$. Here the local private factor includes its port type.

\begin{lemma}[Sweep frame identity]\label{loc:lem:gauge}
Let $O_j$ be the direct sum of the local query actions at cell $j$. In these
coordinates, the complete transducer is
\begin{equation}\label{loc:eq:gauged-transducer}
\begin{split}
S'&=(E_{\rm out}^\dagger\oplus E_{\rm p}^\dagger)
S(E_{\rm in}\oplus E_{\rm p})\\
&=(\mathsf C_M\cdots\mathsf C_1)
\left(I_{\rm pub}\oplus
\bigoplus_jV_{R,j}^\dagger O_jV_{L,j}\right).
\end{split}
\end{equation}
\end{lemma}

\begin{proof}
Before cell $j$, use the public frame $U_L$ on earlier baths and $U_R$ on that cell and
the later baths. Passing through the cell changes only the $j$-th factor, to $U_{L,j}$.
Substitution of \eqref{loc:eq:local-csd} leaves $\mathsf C_j$, followed on private port $j$
by $V_{R,j}^\dagger V_{L,j}$. The intervening public frames cancel.

A factor supported on private port $j$ commutes with $\mathsf C_i$ whenever $i\ne j$, so
move these factors to the right. In the transformed master query, the frame factors in
\eqref{loc:eq:private-frame} outside bath $j$ commute with $O_j$ and cancel as well. What
remains is
\[
V_{R,j}^\dagger V_{L,j}
V_{L,j}^\dagger O_jV_{L,j}
=V_{R,j}^\dagger O_jV_{L,j}.
\]
\end{proof}

Thus the tensor products defining the private frames need not be applied as large
physical operations. Equation \eqref{loc:eq:gauged-transducer} realizes their effect with
gates on the addressed cell.

The same change of coordinates preserves the graph relation \eqref{loc:eq:graph}:
\begin{equation}\label{loc:eq:transformed-graph}
V'=E_{\rm out}^\dagger VE_{\rm in},\qquad
\Gamma'=E_{\rm p}^\dagger\Gamma E_{\rm in},\qquad
P'=E_{\rm in}^\dagger PE_{\rm in}.
\end{equation}
Unitary conjugation of the private block preserves its history-tail norm; padding adds
only directions with zero history amplitude.

One must also transform the source condition. The reflection associated with $P'$ is
\begin{equation}\label{loc:eq:source-reflection}
R_{P'}=E_{\rm in}^\dagger R_PE_{\rm in},
\end{equation}
with the compiler work constrained to its initialized state. Before the compiled $V'$
apply $E_{\rm in}^\dagger$; afterwards apply $E_{\rm out}$ to recover physical
coordinates.

There is one use of the source reflection in rectangular amplification. Thus
\eqref{loc:eq:source-reflection} accounts for two tensor-product frames, and the input and
output account for two more. Further transducer calls do not increase this number. The
remaining bath gates preserve occupation or touch one coherently addressed cell.

\subsection{Inhomogeneous dyadic routing}\label{loc:sec:dyadic}

Fix a bath configuration. On the public rail and private time labels, write $g_j$ for
\[
\begin{split}
g_j\ket{\rm pub}&=c_j\ket{\rm pub}+s_j\ket j,\\
g_j\ket j&=-s_j\ket{\rm pub}+c_j\ket j,
\end{split}
\qquad s_j=\sqrt{1-c_j^2}.
\]
Thus \eqref{loc:eq:gauged-transducer} contains $g_{M-1}\cdots g_0$; dark ports see the
identity. We add identity cells until $M\ge2$ is a power of two.

For a split into adjacent intervals $L,R$, set
\[
\begin{aligned}
r_L&=\prod_{j\in L}c_j,&
r_R&=\prod_{j\in R}c_j,\\
s_L&=\sqrt{1-r_L^2},&
s_R&=\sqrt{1-r_R^2},
\end{aligned}
\]
and $s_*=\sqrt{1-r_L^2r_R^2}$. When $s_*>0$, define
\begin{equation}\label{loc:eq:split-bases}
\begin{split}
M_{\rm in}&=\frac1{s_*}
\begin{pmatrix}r_Rs_L&s_R\\s_R&-r_Rs_L\end{pmatrix},\\
M_{\rm out}&=\frac1{s_*}
\begin{pmatrix}s_L&r_Ls_R\\r_Ls_R&-s_L\end{pmatrix}.
\end{split}
\end{equation}
If $s_*=0$, use identity matrices. Both subintervals then have trivial rotations.

\begin{lemma}[Dyadic factorization]\label{loc:lem:dyadic}
There are time-register unitaries $W_{\rm in},W_{\rm out}$, each with one multiplexed
rotation per binary level, such that
\begin{equation}\label{loc:eq:dyadic}
g_{M-1}\cdots g_0
=\widehat W_{\rm out}\,
\Rot\left(\prod_jc_j\right)_{{\rm pub},0}
\,\widehat W_{\rm in}^\dagger.
\end{equation}
They preserve the bath configuration and act identically on dark ports.
\end{lemma}

\begin{proof}
For one label, the identity is just the definition of $g_0$. Suppose it holds for two
adjacent halves. After reducing each half, the product on the public rail and the two
representatives is
\[
\begin{pmatrix}
r_Rr_L&-r_Rs_L&-s_R\\
s_L&r_L&0\\
s_Rr_L&-s_Rs_L&r_R
\end{pmatrix}.
\]
Multiplying on the left by $1\oplus M_{\rm out}^\dagger$ and on the right by $1\oplus
M_{\rm in}$ gives $\Rot(r_Lr_R)\oplus1$, completing the induction.

At binary level $a$, the two representatives differ in time bit $a$, with every lower
bit zero. The higher bits identify the interval. We can therefore perform the basis
change by rotating bit $a$, controlled on the lower bits and with its angle selected by
the higher prefix and bath configuration.
\end{proof}

If all $c_j$'s agree, \eqref{loc:eq:split-bases} is the homogeneous identity of
\cite{ChenGates}. Allowing them to differ is what lets us use the factorization on
occupied bath configurations.

Suppose at most $K$ bath bins are occupied. All empty bins have the same cell
coefficients and hence the same cosine $c_{\vac}$. The interval product is therefore
\begin{equation}\label{loc:eq:interval-product}
r_I=c_{\vac}^{|I|-k_I}
\prod_{\substack{j\in I\\j\ {\rm occupied}}}c_{\beta_j}.
\end{equation}
Scan the \(K\) records for the occupied factors. For the vacuum contribution, count the
baseline cells of each fixed type and compute the corresponding integer powers; padding
has factor one. The cost is polynomial in \(K,\log M\), precision, and, when variable,
the number of cell types.

In Lemma~\ref{loc:lem:dyadic}, the target bit does not enter the controls. We can therefore
compute flags and rotation data in work registers, rotate the target, and reverse that
computation.

\subsection{Precision of the routing circuit}\label{loc:sec:routing-precision}

An almost inactive interval gives small denominators in the split formula. To control
them, first discretize the elementary angles. This sets a lower bound on every
denominator that is not exactly zero.

Set $c_\beta=\cos\theta_\beta$, with $0\le\theta_\beta\le\pi/2$, and round the
elementary angles to a grid of spacing
\[
\gamma=\Theta(\delta/M),
\]
with both endpoints fixed. Telescoping the $M$ rotations, each perturbed by $O(\gamma)$,
gives $O(\delta)$ total error. This is why we round the angles: a bound from rounding
only sines would deteriorate near $\theta=\pi/2$.

A nonzero rounded angle is at least a constant times $\gamma$, and hence every
nontrivial interval satisfies
\[
\sqrt{1-r_I^2}\ge c\gamma
\]
for an absolute $c>0$. An exactly inactive interval can be recognized from its baseline
type and occupied records, in which case we set $r_I=1$ exactly. The other divisions and
square roots in \eqref{loc:eq:split-bases} then have condition numbers polynomial in
$1/\gamma$.

Evaluate each interval product to absolute error
\[
O(\delta\gamma^3/\log M),
\]
allowing a further constant-power margin for the chosen arithmetic circuit. We compute
integer powers by repeated squaring, taking the local roundoff smaller by a factor $M$.
This requires
\[
b=O\bigl(\log(M/\delta)\bigr).
\]
bits. Approximating each of the $O(\log M)$ multiplexed rotations to error
$O(\delta/\log M)$ gives total error $O(\delta)$, uniformly over valid record
configurations and coherent controls. The arithmetic can be made reversible with
polynomial overhead. On invalid record strings, flag the input and use identity routing.

It follows that a controlled free transducer call, or its adjoint, uses
\begin{equation}\label{loc:eq:router-cost}
\poly\bigl(K,\log M,\log(1/\delta)\bigr)
\end{equation}
gates, plus the addressed local query. To realize the latter for $p$ patch terms,
multiplex their descriptions. From the cell address obtain the interaction and bin,
compute the time, and evaluate the matrices. Apply the controlled encodings and
uncompute the data. Neither address is changed, so the construction acts coherently on
superpositions. It adds polynomial cost in $p$, address length, and precision.

\section{Coherent operations on the compressed bath}\label{loc:sec:bath}

We still need to implement the public frames in \eqref{loc:eq:transformed-graph}. Each is a
product over all time bins, so applying its factors separately would reintroduce a
dependence on the mesh size. We implement each product directly on the compressed bath
register, using a circuit that also works on superpositions of occupied inputs.

We represent a bath configuration by a list of occupied labels $(j,\alpha)$. In each
record, $j$ gives the cell address and $1\le\alpha<d_{\rm b}$ labels the nonvacuum
state. The local bath dimension $d_{\rm b}$ is fixed, and each address $j$ can occur at
most once in a configuration.

\subsection{Two equivalent record encodings}\label{loc:sec:records}

A sorted encoding stores at most $K$ records in increasing address order, with dummy
entries filling the remaining positions. A symmetric encoding uses $K$ slots, each
containing either a dummy $\perp$ or a record. For a configuration with $N\le K$
distinct records $x_1,\ldots,x_N$, define
\begin{equation}\label{loc:eq:symmetric-encoding}
\mathcal I_K\ket{\{x_1,\ldots,x_N\}}
=\frac{\displaystyle\sum_{\pi\ {\rm distinct}}
\ket{\pi(x_1,\ldots,x_N,\perp^{K-N})}}
{\sqrt{K!/(K-N)!}}.
\end{equation}
Different configurations give orthogonal states. The map extends linearly to arbitrary
internal-state superpositions.

Let $P_{\hc}$ reject slot strings with two non-dummy records at the same address.
Pairwise comparisons compute this flag with $O(K^2\log M)$ elementary classical
operations. The valid encoding is the symmetric part of $\ran P_{\hc}$.

\begin{lemma}[Record conversion and addressed access]\label{loc:lem:records}
The sorted and symmetric encodings admit coherent conversion with
$\poly(K,\log M,\log(1/\delta))$ gates to error $\delta$. A unitary on the system and one coherently addressed bath can be applied in the same cost, in addition to its local gate cost, with a unitary guard at occupation $K$.
\end{lemma}

\begin{proof}
For $N$ records, prepare a uniform injection of the ordered labels into $K$ positions.
Choose successive positions from ranges of size $K,K-1,\ldots,K-N+1$. These range states
can be prepared coherently under control of $N$, with polynomial cost in $K$ and the
precision, without storing a list of permutations.

Move each ordered label to its chosen position and fill the rest with dummies. The
output positions and ranks of the distinct addresses determine the injection, so compute
it from the output and erase the choice registers. Each injection has equal amplitude,
giving \eqref{loc:eq:symmetric-encoding}. Reversing the construction returns a valid
symmetric state to the sorted encoding. Repeated dummies are treated as
indistinguishable.

To access address $j$, work in the sorted encoding. Extract its value into a temporary
$d_{\rm b}$-dimensional register, removing the record if present. Apply the local
unitary to the temporary register and the system, and reinsert the output if it is
nonzero. We retain $j$ throughout. The other addresses determine its insertion rank,
allowing all lookup work to be uncomputed.

When the other records already number $K$, use the identity. This condition is
unaffected by the local unitary, so it defines a unitary guard. Otherwise the array has
room for either a vacuum or occupied output, and the guarded operation agrees with the
physical one below the upper occupation boundary.
\end{proof}

The reversible circuits specify an action on the full register space, including invalid
encodings. We bound their approximation error in operator norm on the valid subspace.

\subsection{The generator of a product of weak basis changes}\label{loc:sec:frame-generator}

Consider a public frame
\[
F(h)=\bigotimes_{j=1}^M U_j(h)
\]
satisfying \eqref{loc:eq:public-frame-weak}. First separate the zero-coupling factors
$U_j(0)$. They fix the vacuum and change only the internal states of occupied records,
so applying them costs polynomially many gates in $K$ and $\log M$.

For the remaining factors, take the small Hermitian logarithm
\[
U_j(0)^\dagger U_j(h)=e^{-ih_j},\qquad
\norm{h_j}=O(\sqrt h).
\]
After decreasing $h_0$, these logarithms can all be evaluated by the uniformly
convergent power series around the identity. Since different factors act on distinct
bins, their product is the exponential of $\sum_jh_j$. For a forward frame, apply this
weak rotation first, followed by the zero-coupling factors.

Write each local generator as a scalar, a creation column, and an occupied block:
\[
a_j=\bra0h_j\ket0,\qquad
b_j=P_j^1h_j\ket0,\qquad
q_j=P_j^1h_jP_j^1-a_jP_j^1.
\]
The scalar $a_{\rm tot}=\sum_ja_j$ contributes a known phase, which we retain also for a
controlled frame operation. Define
\begin{equation}\label{loc:eq:collective-column}
\ket b=\sum_{j,\alpha}(b_j)_\alpha\ket{j,\alpha},
\qquad
B=\norm b=\left(\sum_j\norm{b_j}^2\right)^{1/2}.
\end{equation}
For repeated local cells,
\[
B=O(\sqrt{Mh})=O(\sqrt\Omega).
\]
and the state $b/B$ can be prepared from the constant-size internal column and a uniform
time label. Several fixed cell types can be handled by including the $O(p)$ type data in
this preparation.

Let $\widehat N$ count non-dummy slots, and set
\[
f(N)=
\begin{cases}
(K-N)^{-1/2},&N<K,\\
0,&N=K.
\end{cases}
\]
Define
\begin{equation}\label{loc:eq:slot-creation}
T=P_{\hc}
\left(\sum_{p=1}^K\ket b\bra\perp_p\right)
f(\widehat N)P_{\hc},
\end{equation}
where the one-slot operator is placed on slot $p$. Let $q$ be the one-record operator
with blocks $q_j$ and zero on the dummy. Put
\begin{equation}\label{loc:eq:slot-generator}
\widehat H_K=T+T^\dagger+
P_{\hc}\left(\sum_{p=1}^Kq^{(p)}\right)P_{\hc}.
\end{equation}

\begin{lemma}[Exact compressed generator]\label{loc:lem:slot-generator}
On the valid encoded subspace,
\begin{equation}\label{loc:eq:generator-intertwining}
\widehat H_K\mathcal I_K
=\mathcal I_K
P_{\le K}\left(\sum_jh_j-a_{\rm tot}I\right)P_{\le K}.
\end{equation}
The valid symmetric subspace is invariant under $\widehat H_K$.
\end{lemma}

\begin{proof}
There are $K-N$ dummy slots when the input has $N$ records. Summing slot creation maps
therefore multiplies a symmetric matrix element by $\sqrt{K-N}$. In
\eqref{loc:eq:slot-creation}, the number function acts first and divides out this factor.
What remains is the entry of $b_j$ at an empty address. The $P_{\hc}$ projection kills
occupied targets, and at $N=K$ the creation map is zero.

Adjointing gives annihilation, now with the number function after the annihilation
operator as in \eqref{loc:eq:slot-creation}. The $q_j$ sum in \eqref{loc:eq:slot-generator} acts
once on every record. These are exactly the projected physical matrix elements. They
commute with slot permutations, and the outer projectors exclude repeated addresses,
proving invariance of the valid symmetric subspace.
\end{proof}

The restriction to one record per address does not remove the dummy-slot multiplicity.
Without the number function, a vacuum-to-record matrix element would still acquire the
square root of the number of empty slots.

\subsection{Gates for a compressed frame}\label{loc:sec:frame-gates}

We now block encode the terms of \eqref{loc:eq:slot-generator}. Prepare $b/B$ and apply the
dummy projector to encode
$\ket b\bra\perp/B$ with normalization one. Since $f(\widehat N)$ is bounded by one, a number-controlled rotation encodes it. Computing the validity flag gives an encoding of $P_{\hc}$.

Taking an LCU over the $K$ slots, their adjoints, and the one-record $q$ blocks gives
\begin{equation}\label{loc:eq:frame-normalization}
\alpha_{\rm enc}\le K(2B+D_0),\qquad
D_0\ge\max_j\norm{q_j},
\end{equation}
with gate cost polynomial in $K,p,\log M$ and the bit precision. We may use a fixed
upper bound for $D_0$. The normalization depends on the column norm $B$; replacing it by
$\sum_j\norm{b_j}$ would introduce a positive power of the mesh size.

Apply Hamiltonian simulation by qubitization \cite{LowChuang} to this Hermitian block
encoding. For unit time and error $\delta$, its cost is polynomial in
\begin{equation}\label{loc:eq:frame-cost}
K,\ p,\ 1+\Omega,\ \log M,\ \log(1/\delta).
\end{equation}
The projector and LCU circuits act on the full $K$-slot register space and may leave the
valid record subspace internally. On a valid input, \eqref{loc:eq:generator-intertwining}
identifies the target evolution, while the Hamiltonian-simulation error bounds the
failure to return to that subspace.

Use local logarithms accurate to $O(\delta/M)$, requiring $O(\log(M/\delta))$ bits.
Their errors sum to $O(\delta)$ in the generator; use this precision for the known
scalar phase as well. There are two cases for creation. If $B\sqrt K$ is below its error
allowance, omit the term. Otherwise the normalized column $b$ can be prepared with
condition cost polynomial in $K/\delta$. Thus no additional inverse power of the small
coupling is charged.

Controls and adjoints use the same circuits. For $F(h)^\dagger$, reverse the compressed
weak rotation and the zero-coupling factors. The following occupation estimate covers
both orders.

\section{Uniform cutoffs and the total gate count}\label{loc:sec:assembly}

The physical tail estimate \eqref{loc:eq:physical-tail} bounds the occupation entering each
regional factor. During its compiled implementation, however, queries and changes of
basis may create further records. We need a cutoff that also controls these intermediate
states.

\subsection{A cutoff for the compiled program}\label{loc:sec:compiled-cutoff}

Consider a circuit with three kinds of operations: unitaries preserving occupation, $Q$
unitaries changing occupation by at most one, and collective rotations of the form
treated in Section~\ref{loc:sec:frame-generator}. Let $B_s$ be the norm
\eqref{loc:eq:collective-column} for rotation $s$, and set $B_{\rm tot}=\sum_sB_s$. The
input has occupation at most $N_0$.

Define the ideal cutoff circuit $V_K$ in physical configuration space as follows.
Compress each collective generator to $P_{\le K}$ and put the guard of
Lemma~\ref{loc:lem:records} on each addressed operation.

\begin{lemma}[Compiled-circuit cutoff]\label{loc:lem:cutoff}
For $K\ge N_0$,
\begin{equation}\label{loc:eq:cutoff-bound}
\begin{aligned}
&\norm{(V-V_K)P_{\le N_0}}\\
&\quad\le
\bigl(B_{\rm tot}\sqrt{K+1}+2Q\bigr)\\
&\qquad\cdot
e^{N_0+Q\log5+2\sinh(1)B_{\rm tot}\sqrt K-K}.
\end{aligned}
\end{equation}
In particular, error at most $\delta$ follows from
\begin{equation}\label{loc:eq:cutoff-choice}
K=O\bigl(1+N_0+Q+B_{\rm tot}^2+\log(1/\delta)\bigr).
\end{equation}
The estimate allows arbitrary coherent controls and an external reference.
\end{lemma}

\begin{proof}
For a physical creation operator with column $b_s$,
\[
\norm{T_sP_N}\le B_s\sqrt{N+1}.
\]
To see this, embed the configurations with distinct addresses into bosonic Fock space
and use the creation-operator norm on the $N$-particle sector. Projecting back can only
decrease the norm. On the cutoff space, the creation part consequently has norm at most
$B_s\sqrt K$.

Conjugation by $e^{\cN}$ multiplies creation by $e$ and annihilation by $e^{-1}$. The
difference between the conjugated Hermitian generator and the original one has norm at
most
$2\sinh(1)B_s\sqrt K$. Duhamel's formula, or the interaction picture with the Hermitian part, gives
\[
\norm{e^{\cN}e^{-iH_{s,K}t}e^{-\cN}}
\le e^{2\sinh(1)B_s\sqrt K\,t}.
\]

A unitary that changes occupation by at most one has three number diagonals. Each is a
Fourier coefficient of
$e^{i\theta\cN}Ue^{-i\theta\cN}$
and hence has norm at most one. Its weighted norm is therefore bounded by
$e+1+e^{-1}<5$. An occupation-preserving unitary has weighted norm one. The guarded operations remain unitary with the same number bandwidth, so these estimates apply to them too.

For every truncated prefix, including times inside a collective rotation, we obtain
\begin{equation}\label{loc:eq:weighted-prefix}
\norm{e^{\cN}V_{K,{\rm pre}}P_{\le N_0}}
\le e^{N_0+Q\log5+2\sinh(1)B_{\rm tot}\sqrt K}.
\end{equation}
Its amplitude on the top sector is bounded by an additional factor $e^{-K}$.

Use Duhamel's formula to compare a full collective rotation with its compressed
generator. The coupling out of $P_{\le K}$ is supported on the top sector and has norm
at most $B_s\sqrt{K+1}$. A full addressed unitary also differs from its guard only on
that sector, by norm at most two. Telescope through the circuit, keeping full future
unitaries on one side and truncated prefixes on the other. The preceding top-sector
estimate then gives \eqref{loc:eq:cutoff-bound}.

Finally,
\[
2\sinh(1)B_{\rm tot}\sqrt K
\le\tfrac14K+4\sinh^2(1)B_{\rm tot}^2.
\]
The logarithm of the prefactor in \eqref{loc:eq:cutoff-bound} can be absorbed by another
fixed fraction of $K$ and the terms in \eqref{loc:eq:cutoff-choice}. Increasing its constant
proves the final assertion.
\end{proof}

The cutoff bounds the logical bath occupation. Arithmetic work and block-encoding
ancillas are supplied by their own unitary circuits, whose finite synthesis errors will
be added separately. They require no further physical time-bin registers.

\subsection{One regional circuit}\label{loc:sec:regional-circuit}

Consider a sweep with at most $p$ local terms and $M\le pJ$ physical cells before
padding. On a segment of length $t\le{t_\star}$, its action is $\Omega=O(pt)$.
Lemma~\ref{loc:lem:sparse} approximates it on $P_{\le N_0}$ using
\[
Q=O\bigl(1+N_0+\Omega+\log(1/\delta)\bigr).
\]
Implement the transformed transducer \eqref{loc:eq:gauged-transducer} by the preceding
routing construction. The free routing preserves occupation, and the addressed query,
including its local private-frame factors, changes it by at most one.

The reuse and coefficient registers leave the bath occupation unchanged. This leaves the
four collective public-frame factors described after \eqref{loc:eq:source-reflection}. Each
has $B_s=O(\sqrt\Omega)$, giving $B_{\rm tot}^2=O(\Omega)$.

Take a precision parameter
\[
b\ge C\bigl(1+\log(2M)+\log(1/\delta)\bigr).
\]
with its constant increased to cover time-address arithmetic and the allocation of error
among subroutines. We can pad the record capacity to
\begin{equation}\label{loc:eq:regional-capacity}
K=\Theta(1+N_0+p+b).
\end{equation}
which suffices for Lemma~\ref{loc:lem:cutoff} and gives $Q=O(K)$, $p\le K$, and $b=O(K)$.

Let $a\ge1$ bound the polynomial cost of the reversible $b$-bit arithmetic and
local-matrix evaluation. Choose $a$ large enough to include fixed-dimensional gate
synthesis, coefficient preparation, and elementary functions as well. The preceding
constructions permit a fixed choice, independent of $n,J,p,T,\eps$.

\begin{proposition}[Regional implementation]\label{loc:prop:regional}
A forward or backward regional sweep can be approximated on $P_{\le N_0}$ to operator
error $\delta$, with additional work reset in the target, using
\begin{equation}\label{loc:eq:regional-gates-explicit}
G_{\rm reg}=O\bigl(K^3(1+\sqrt p)b^a\bigr),
\qquad
W_{\rm reg}=O(K^2b^a).
\end{equation}
Here $K$ is chosen as in \eqref{loc:eq:regional-capacity}. The gate bound is before imposing
geometric routing.
\end{proposition}

\begin{proof}
Apply Lemma~\ref{loc:lem:sparse} with a constant fraction of the error budget. The frame
identity preserves the history-tail bound on the transformed source subspace, and
Lemma~\ref{loc:lem:cutoff} bounds the error of truncating to the capacity in
\eqref{loc:eq:regional-capacity}. It remains to count the gates and allocate the synthesis
error.

A routing call uses $O(\log M)$ dyadic levels. At each level, scan the records, count
the baseline bins, and evaluate the integer powers. These steps cost $O(Kb^a)$ gates
after increasing $a$. Since $\log M\le b=O(K)$, one routing call costs $O(K^2b^a)$,
including addressed extraction and reinsertion. Selecting among $p$ local descriptions,
evaluating the sampled matrices, and applying their constant-size encodings adds
$O(pb^a)$ gates. Over $Q=O(K)$ calls, the total is $O(K^3b^a)$.

For each collective frame, \eqref{loc:eq:frame-normalization} gives block-encoding
normalization $O(K(1+\sqrt p))$. One use costs $O(K^2b^a)$ gates: compute the
duplicate-address flag by pairwise comparisons, perform reversible slot selection and
number arithmetic, and use the one-record preparations and blocks of
Section~\ref{loc:sec:frame-gates}. Qubitization to error $O(\delta)$ takes $O(K(1+\sqrt
p)+b)$ uses. The four frames therefore satisfy the gate bound in
\eqref{loc:eq:regional-gates-explicit}.

Record conversion fits within $O(K^3b^a)$ gates. Choose successive positions from those
remaining, route the $K$ records by reversible swaps, and recover the choices from the
output ranks. Repeated scans and comparisons implement this without a permutation table
or factorial-size register. The records and comparisons use $O(K^2b)$ qubits, while the
supplied evaluation circuit and arithmetic may require $O(b^a)$ additional work. Thus
$O(K^2b^a)$ qubits suffice for both, including reuse indices and Hamiltonian-simulation
controls.

Give each repeated subroutine error $O(\delta/Q)$ and each collective frame a constant
fraction of $\delta$. The conditioning bounds in Section~\ref{loc:sec:routing-precision} and
the local-frame calculation require $O(b)$ precision bits, with constants enlarged as
needed. Adding the errors by a unitary hybrid makes the total from compilation, cutoff,
and synthesis at most $\delta$.

For a backward regional factor, apply the compiler to the adjoint cells in reverse order
on their own sparse input subspace. The sampled matrices are those of the forward sweep,
read in reverse chronological order. Its inactive block is the adjoint of
\eqref{loc:eq:dcell}, which interchanges the singular frames. All the preceding costs and
error estimates are unchanged.
\end{proof}

To implement the physical source reflection, count the occupied records and test the
initialized compiler work. Conjugating it uses the two frame factors in
\eqref{loc:eq:source-reflection}, with their controlled scalar phases retained. These
operations are included in Proposition~\ref{loc:prop:regional}.

\subsection{A common spatial encoding}\label{loc:sec:spatial-encoding}

Store the $J$ baths of each physical interaction label in a record array. While a
regional circuit runs, merge the arrays in its patch and attach the local interaction
label to each record. Sorting, merging, and conversion between the two encodings have
the costs given above. The labels specify where the records must return, so the merge
can be undone coherently.

Give each interaction array the largest regional capacity. If a region exceeds its
allowed total input capacity, use an identity extension. On a valid branch, the guarded
logical circuit stays within the capacity and its flag can be uncomputed; finite
synthesis errors are charged separately. After applying the regional output frame,
return the records to their physical arrays. The next layer then acts on the same bath
in the same basis.

For the work registers, partition the lattice into the elementary boxes formed by all
cuts in all coordinates. A regional box contains at most $2^D$ such boxes. Give each
elementary box a workspace pool of the largest size required by
\eqref{loc:eq:regional-gates-explicit}. A regional circuit uses only constantly many pools,
together with $O(pKb)$ bath-array qubits. As $p\le K$, the complete patch occupies
$O(K^2b^a)$ qubits. There are at most $n$ elementary boxes, which also bounds the total
workspace allocation.

Arrange the registers of each elementary box in a contiguous local block. A path within
a regional patch implements any two-qubit gate with $O(K^2b^a)$ neighboring swaps and
gates. Swap the registers back afterwards to preserve their locations for the next
layer. A sufficient geometric gate bound is therefore
\begin{equation}\label{loc:eq:geometric-regional}
G_{\rm reg}^{\rm geom}
=O\bigl(K^5(1+\sqrt p)b^{2a}\bigr).
\end{equation}
This also bounds the sequential depth of a patch circuit. The disjoint patches in each
layer run in parallel.

For periodic coordinates on hardware with open boundaries, fold each coordinate into two
sheets over an interval. Periodic neighbors then occupy neighboring or coincident sheet
positions. This increases local site density by at most $2^D$ and preserves bounded
interaction range \cite{HHKL}, allowing the same layout.

\subsection{Proof of Theorem~\ref{loc:thm:main}}\label{loc:sec:main-proof}

Fix the short time ${t_\star}$ chosen after Proposition~\ref{loc:prop:layers}, and put
\[
R=\max\{1,\lceil T/{t_\star}\rceil\}.
\]
For $T=0$ or empty $E$, the identity channel suffices. More generally,
\[
\norm{\mathcal E(T,0)-\id}_\diamond\le2gn\Lambda_{\rm loc} T.
\]
so we can also use the identity whenever this bound is at most $\eps$. Otherwise
$\log(1/T)$ is at most $O(\log(n/\eps))$, which will bound the precision needed for
small times. Divide $[0,T]$ into $R$ equal segments of length $t=T/R$, with fresh baths
for each segment. Using the same power-of-two $J$ on every segment gives a uniform mesh
of width $h=T/(RJ)$.

By Lemma~\ref{loc:lem:temporal-mesh}, the temporal discretization error is at most $\eps/4$
for sufficiently large $J$. For a fixed constant $C$, it suffices to choose
\begin{equation}\label{loc:eq:mesh-choice}
\begin{split}
J\ge C\biggl[&
1+\frac{{t_\star}}{h_0}
+\frac{n^2(T+1)}{\eps}
+\frac{nB}{\eps}\\
&+\left(\frac{nK_t(T+1)}{\eps}\right)^{1/\alpha}
\biggr].
\end{split}
\end{equation}
Take the smallest power of two satisfying this bound. Then
\[
\log J=O(\alpha^{-1}\log X)=O(\log X)
\]
since the H\"older exponent is fixed. The small-time estimate above also gives
$\log(1/h)=O(\log X)$, covering the precision required for the local free frames.

On each segment, apply Proposition~\ref{loc:prop:layers} with joint error $\eta=\eps/(16R)$.
We may take
\[
\ell=O(\log X),\qquad \ell>\rho.
\]
leaving shorter coordinates unsubdivided. The resulting word has $r_0=3^D$ layers, with at most
\[
p=O(\ell^D)=O(\log^D X)
\]
sites and local terms in each regional factor. Across the full evolution, the number of
nonidentity factors is bounded by
\[
L_{\rm tot}=O(r_0Rn).
\]
This count allows unequal box volumes and is sufficient for the gate bound.

Let $b=C'(1+\log X+\log J)$. Choose the per-interaction physical cutoff $k$ large enough
that $k\ge b$ and that \eqref{loc:eq:regional-tail}, summed over all regional inputs and
multiplied by two, is at most $\eps/16$. Equation \eqref{loc:eq:physical-tail} gives such a
choice with
\[
k=O(b),\qquad N_0=pk=O(pb),\qquad
\Omega=O(p).
\]
Assign regional approximation error
$\delta=\eps/(16L_{\rm tot})$.
Its logarithm is $O(b)$. Proposition~\ref{loc:prop:regional} then permits
\begin{equation}\label{loc:eq:global-capacities}
Q=O(pb),\qquad K=O(pb),\qquad b=O(\log X).
\end{equation}

Now telescope the approximate regional circuits against the exact physical word. At each
input, the exact prefix satisfies \eqref{loc:eq:regional-tail}. The regional error on its
good component is at most $\delta$; on the complement, the difference of two isometries
has norm at most two. Summing over regional factors and segments gives
\[
L_{\rm tot}\delta+
2\sum_{\text{regional inputs}}
\norm{(I-P_{\rm good})V_{\rm pre}P_{\vac}}
\le\eps/8.
\]
The occupation estimate is applied to exact prefixes throughout this hybrid, so
approximate intermediate states need not have exact sparse support.

For isometries with a common environment,
\[
\norm{\Tr_E(V\,\cdot\,V^\dagger)
-\Tr_E(W\,\cdot\,W^\dagger)}_\diamond
\le2\norm{V-W}.
\]
The spatial joint error summed over segments is at most $\eps/16$. Adding the regional
implementation error and temporal channel error gives total diamond error at most
\[
2(\eps/8+\eps/16)+\eps/4<\eps.
\]

For the resource count, substitute $K=O(pb)$ into \eqref{loc:eq:regional-gates-explicit} and
\eqref{loc:eq:geometric-regional}. For $p,b\ge1$,
\[
G_{\rm reg}=O(p^4b^{a+3}),\qquad
G_{\rm reg}^{\rm geom}=O(p^6b^{2a+5}).
\]
There are $L_{\rm tot}=O(n(T+1))$ regional factors and $p=O(b^D)$, giving
\begin{equation}\label{loc:eq:explicit-resources}
\begin{split}
G_{\rm geom}&=O\bigl(n(T+1)b^{6D+2a+5}\bigr),\\
\operatorname{depth}&=O\bigl((T+1)b^{6D+2a+5}\bigr),\\
\operatorname{space}&=O\bigl(nb^{2D+a+2}\bigr).
\end{split}
\end{equation}
Run disjoint patches in parallel through the $3^D$ layers to obtain the depth bound. The
common workspace pools and $W_{\rm reg}=O(p^2b^{a+2})$ give the space bound; discarding
the bath after each segment removes a factor $T$ from peak space. This proves
\eqref{loc:eq:explicit-resources} and the theorem.

These exponents come from the specified reversible routines and have not been optimized.
The fine mesh enters only through the length of its binary address. In particular,
temporal variation affects the gate count through $\log J$, while the bath action within
a short patch remains $O(p)$.

\subsection{An effective temporal modulus}\label{loc:sec:effective-modulus}

H\"older regularity was used to choose a mesh with logarithmic address length. A
different effective regularity bound can be used in the same construction, provided we
include the resulting address length in the cost.

\begin{corollary}[General freezing modulus]\label{loc:cor:modulus}
Suppose the bounded local generator has an effective bound
\[
\int_0^T\norm{\cL(s)-\cL_h(s)}_\diamond\,ds\le F(h)
\]
for midpoint freezing on a uniform mesh. Choose a power-of-two $J$ on each of the $R$
equal segments such that, for $h=T/(RJ)$,
\[
F(h)+Cn^2Th\le\eps/4,\qquad h\le h_0.
\]
For coherently evaluable local matrices as above, the gate count is
\[
n(T+1)\poly\bigl(\log(n(T+1)/\eps)+\log J\bigr)
\]
and \eqref{loc:eq:explicit-resources} gives the corresponding depth and peak space.
\end{corollary}

\begin{proof}
The temporal error is at most \(F(h)+Cn^2Th\le\eps/4\), by variation of constants and
the local-cell/Lie-product bound. Thereafter the proof uses finite-bin norm estimates
only. Reuse its compiler and resource count with
\(b=O(1+\log(n(T+1)/\eps)+\log J)\)
with \(\ell,k=O(b)\). The stated polynomials follow.
\end{proof}

For the H\"older class, Theorem~\ref{loc:thm:main} supplies \(\log J=O(\log X)\). The more
general statement must retain \(\log J\): a weaker modulus may force a finer mesh.

\section{Conclusion}\label{sec:conclusion}

Causal query compression gives a mesh-independent query bound for Lindblad evolution
under coherent access to its Hamiltonian and individual jump operators. The bound is
worst-case optimal for \(\tau\ge1\). Its proof uses a chronological query history,
weak entry and clean-return amplitudes, and a two-length reuse construction that removes
the auxiliary input. Retaining intervention records also gives the oracle-action
characterization for the bounded adaptive Markovian model.

For efficiently evaluable local generators, the spatial and occupation estimates turn
this query construction into an elementary-gate implementation. The common bath must
be retained across regional factors, and the occupation cutoff must remain valid during
queries and collective frame changes. Under these conditions,
\eqref{loc:eq:explicit-resources} gives nearly linear dependence on \(n(T+1)\),
with polylogarithmic dependence on precision and temporal variation. The stated
polynomial exponents include arithmetic and geometric routing and have not been optimized.

The general oracle theorem does not by itself supply such a gate bound: the local
implementation uses fixed spatial parameters, common normalizations, and efficient
coherent evaluation. Corollary~\ref{loc:cor:modulus} states the remaining mesh-address
cost for weaker effective temporal regularity. Improving the explicit routing overheads,
or replacing the short step by a higher-order circuit with verified weak-cell bounds,
would refine these implementation costs.

\section*{Acknowledgment}
AI-assisted tools were used for limited editorial support in refining the prose.
The mathematical arguments, results, and conclusions are solely those of the author.

\section*{Author Declarations}
The authors have no conflicts to disclose.

\end{document}